\documentclass[reqno]{amsart}

\usepackage{amssymb}
\usepackage{hyperref}
\usepackage{amsfonts} 
\usepackage{amsmath}
\usepackage{color,ulem}
\numberwithin{equation}{section}
\newtheorem{theorem}{Theorem}[section]
\newtheorem{proposition}[theorem]{Proposition}
\newtheorem{lemma}[theorem]{Lemma}

\newcommand{\abs}[1]{ | #1 |}

\newcommand{\tnorm}[1]{\left|\!\left|\!\left| #1 \right|\!\right|\!\right|}

\renewcommand\H{\mathcal{H}}

\newcommand\R{\mathbb R}

\newcommand\Z{\mathbb Z}
\newcommand\K{\mathcal{K}}
\newcommand\D{\mathcal{D}}
\newcommand\Db{\mathbf{D}}
\newcommand\G{\mathcal{G}}
\newcommand\U{\mathcal{U}}
\newcommand\Ll{\mathcal{L}}
\newcommand\A{\mathbf{A}}

\newcommand\x{\mathbf{x}}

\newcommand\di{\mathrm d}

\newcommand{\Hc}{\mathcal{H}_{\mathrm{c}}}
\newcommand{\Cc}{C_{\mathrm{c}}}

\newcommand\El{\mathbf E}

\newcommand{\T}{\mathcal{T}}
\newcommand\F{\mathbf{F}}

\renewcommand\S{\mathcal{S}}
\newcommand\tr{\mathrm{tr}}

\newcommand\e{\mathrm{e}}
\newcommand{\la}{\langle}
\newcommand{\ra}{\rangle}
\renewcommand\P{\mathbb P}
\newcommand\E{\mathbb E}

\newcommand{\HH}{\mathbb{H}}
\newcommand{\LL}{\mathbb{L}}
\newcommand{\cJ}{\mathcal{J}}

\newcommand\I{\mathcal{I}}

\newcommand{\eq}[1]{\eqref{#1}}

\begin{document}
 \normalem

\title[The Liouville Equation for  Magnetic Schr\"odinger
operators]{The Liouville Equation for Singular Ergodic Magnetic Schr\"odinger
operators}

\author[Y. Kang]{Yang Kang}
\address[Kang]{Michigan State University, Department of Mathematics, East Lansing, MI 48823, USA}
 \email{ykang@math.msu.edu}

\author[A. Klein]{Abel Klein}
\address[Klein]{University of California, Irvine,
Department of Mathematics,
Irvine, CA 92697-3875,  USA}
 \email{aklein@math.uci.edu}

\thanks{A.K was  supported in part by NSF Grant DMS-0457474.}

%%%%%%%%%%%%%%%%%%%%%%%%%%%%%%%%%%%%%%%%%%%%%%%%%%%%

\begin{abstract}
We  study the time evolution of a density matrix in a  quantum mechanical system described
by an ergodic magnetic  Schr\"odinger operator with singular magnetic and electric potentials, the electric field being introduced adiabatically. We construct a unitary propagator that solves weakly the corresponding time-dependent Schr\"odinger equation, and  solve a Liouville equation in an appropriate Hilbert space. 
\end{abstract}

%%%%%%%%%%%%%%%%%%%%%%%%%%%%%%%%%%%%%%%%%%%%%%%%%%%%
\maketitle

%\tableofcontents

\section{Introduction}\label{secintro}

We study non-interacting quantum particles in a disordered background described by a one-particle ergodic magnetic Schr\"odinger operator. The system is taken to be at equilibrium at time $t=-\infty$ in a state given by a one-particle density matrix, and
 the electric field is introduced adiabatically.  The time evolution of the density matrix  is then described  by a Liouville equation.   In this article we consider singular ergodic magnetic Schr\"odinger operators, the conditions on the magnetic and electric potentials  only ensure that    $\Cc^\infty(\R^d)$ is  a form core for the  magnetic Schr\"odinger operator, and prove that the   Liouville equation can be  given a precise meaning and solved in an appropriate Hilbert space. A similar result was  previously obtained  by Bouclet, Germinet, Klein and Schenker  \cite{BGKS} under stronger conditions on  the magnetic and electric potentials that yield essential self-adjointness  of the magnetic Schr\"odinger operator on  $\Cc^\infty(\R^d)$.

We consider  magnetic  Schr\"odinger operators of the form 
\begin{align}  \label{H(A,V)4} 
  H=    H(\A,V) :=  \left(-i\nabla - \A \right)^2 + V   \; \; \;
\mathrm{on} \; \; \;
\H :=\mathrm{L}^2(\mathbb{R}^d),
\end{align}
where the magnetic potential $\A$ and the electric potential $V$ satisfy:     
\begin{itemize}
\item[(i)]  $\A  \in \mathrm{L}^2_{\mathrm{loc}}(\R^d; \R^d)$, 
\item[(ii)]  $V= V_+ - V_-$, where
 $V_\pm \in \mathrm{L}^1_{\mathrm{loc}}(\R^d)$, $V_\pm \ge 0$,
 and $ V_-$  is relatively form bounded with respect to $\Delta$ with relative bound $<1$, i.e., there are $0 \le\alpha  < 1$ and $\beta \ge 0$ such that
\begin{equation}\label{form_bound}
 \la \psi, V_-\psi \ra  \le \alpha \la \psi,-\Delta \psi\ra + \beta ||\psi ||^2.
\end{equation}
\end{itemize}
$H$ is naturally defined as a semi-bounded self-adjoint operator  by a quadratic form, with $C_c^{\infty} (\R^{d})$ being a form core, and    the diamagnetic inequality  holds for $H$ (cf. \cite[Theorems~2.2 and 2.3]{Si79}; although $V_{-}=0$ in \cite{Si79}, the results extends  to $V_{-}$ relatively bounded as in \eq{form_bound}  by an approximation argument as in
 \cite[Proposition~7.7 and Theorem~7.9]{F}.). The  usual trace estimates for Schr\"odinger operators hold for $H$ (cf.  \cite[Proposition~2.1]{BGKS}).

In a disordered background the system is modeled by an  ergodic magnetic Schr\"odinger operator 
 \begin{align}  \label{H(A,V)omega} 
  H_{\omega}=    H(\A_{\omega},V_{\omega}) :=  \left(-i\nabla - \A_{\omega} \right)^2 + V_{\omega}   \; \; \;
\mathrm{on} \; \; \;
\mathrm{L}^2(\mathbb{R}^d),
\end{align}
where the  parameter $\omega$ runs in a probability space $(\Omega,\P)$, and for $\P$-a.e.\ $\omega$ we assign a  magnetic potential $A_\omega$ and an
electric potential $V_\omega$ such that  $H_{\omega}=    H(\A_{\omega},V_{\omega}) $  is  as  in \eq{H(A,V)4}.    The ergodic system satisfies a covariance relation:  there exist an ergodic group $\{\tau(a); \ a \in \Z^d\}$ of
measure preserving transformations on the probability space
$(\Omega, \P)$  and     a unitary
projective representation   $U(a)$ of $\Z^d$ on $\mathrm{L}^2(\mathbb{R}^d)$ 
such that for $\P$-a.e.\ $\omega$  
\begin{align} \label{covintro} 
U(a)  H_\omega (t) U(a)^* =  H_{\tau(a) \omega}(t) \;\; \text{and}\; \; 
U(a) \chi_b U(a)^* &=\chi_{b+a} \; \; \text{for     all  $a ,b \in \Z^d$,}
\end{align}   
where $\chi_a$ denotes  multiplication by the characteristic
function of the unit cube centered at $a$. 
It follows from ergodicity that $ (V_\omega)_-$ satisfies \eqref{form_bound} $\P$-a.e.\ with the same constants $\alpha$ and $\beta$, and there exists a constant $\gamma$ such that 
\begin{align}
   H_\omega + \gamma   \ge 1 .   \label{Homegamma}  
\end{align}

 At time $t=-\infty$, the system is  in equilibrium 
in the state given by a one-particle density matrix ${\zeta}_\omega=f(H_\omega)$,
where $f$ is a non-negative  bounded function with fast enough decay at infinity. At
zero temperature we take \ 
${\zeta}_\omega=P_\omega^{(E_F)}=\chi_{(-\infty,E_F]}(H_\omega)$, the Fermi
projection corresponding to the Fermi energy $E_F$. 
A homogeneous  electric field $\El$ is then slowly switched on from time $t= -\infty$ to time $t=0$.  We take  $\eta>0$, let    $t_- = \min \, \{t, 0\}$,
$t_+ = \max \, \{t, 0\}$, and set
\begin{equation} \label{defEintro}
\mathbf{E}(t)=
\mathrm{e}^{\eta t_-}\mathbf{E}\, .
\end{equation}
The dynamics are now generated, in the appropriate gauge (see \cite[Section~2.2]{BGKS}),  by the ergodic
time-dependent  magnetic Schr\"odinger operator
\begin{equation}\label{eq:Homegaintro}
    H_\omega(t) = ({-i\nabla} - \A_\omega - \F(t))^2  + V_\omega  =
G(t) H_\omega G(t)^\ast \, ,
\end{equation}
where
\begin{equation} \label{Ftintro}
\F(t) =  \int_{-\infty}^t \mathbf{E}(s) \di s =
\left(\textstyle{\frac{\mathrm{e}^{\eta t_-}} \eta } + t_+  \right) \El \, ,
\end{equation}
and $G(t)=\mathrm{e}^{i\F(t)\cdot x}$ is a gauge transformation on
$\mathrm{L}^2(\R^d)$.  Note that $  H_\omega(t)$  is almost surely   a magnetic  Schr\"odinger operator as  in \eq{H(A,V)4}. 
Under this time evolution, $\varrho_\omega(t)$, the
 density matrix  at time $t$,   is  the solution of
 the Liouville equation  given formally by
\begin{equation}\label{Liouvilleeq}
\left\{
\begin{array}{l}i\partial_t \varrho_\omega(t) = [H_\omega(t),\varrho_\omega(t)]
\\
\lim_{t \to  -\infty} \varrho_\omega(t)= {\zeta}_\omega
\end{array}
\right.  \, .
\end{equation}

We will  give a  precise meaning to this Liouville equation for  an ergodic  magnetic Schr\"odinger operator $H_{\omega}$ as in \eq{H(A,V)omega}, and construct its solution. We will assume
\begin{equation}  \label{key_hypothesis}
\E\left\{ \left\| x_k^{2} \,{\zeta}_\omega \chi_0\right \|_2^2\right\} < \infty \quad \text{for}\quad k=1,\cdots,d,
\end{equation}
 where  $\|S\|_2$ denotes the Hilbert-Schmidt norm of
the operator $S$. 
(If ${\zeta}_\omega = P^{(E_F)}_\omega$ where $E_F$ falls inside a
gap of the spectrum of $H_\omega$, or ${\zeta}_\omega=f(H_\omega) $ with $f$
smooth and appropriately decaying at high energies, then \eqref{key_hypothesis} is
readily fulfilled by general arguments, cf.  \cite{GKdecay}. It also holds for  ${\zeta}_\omega = P^{(E_F)}_\omega$  if 
 the Fermi energy $E_F$ is inside a region of complete localization by \cite{GKsudec}.)

 The Liouville equation \eqref{Liouvilleeq} was studied by Bouclet, Germinet, Klein, and Schenker \cite{BGKS} 
 under the stronger assumption that the magnetic and electric potentials satisfy the Leinfelder-Simader conditions \cite{LS}:
\begin{itemize}
\item[(a)]   $\A  \in \mathrm{L}^4_{\mathrm{loc}}(\R^d; \R^d)$  with
$\nabla \cdot \A  \in \mathrm{L}^2_{\mathrm{loc}}(\R^d)$.

\item[(b)] $V = V_+  - V_- $ with
 $V_\pm  \in \mathrm{L}^2_{\mathrm{loc}}(\R^d)$, $V_\pm  \ge 0$,
 and
$ V_-$  relatively bounded with respect to
$\Delta$ with relative bound $<1$.
\end{itemize}
Under these conditions  $ H(\A,V)$
 is essentially self-adjoint on $\Cc^\infty(\R^d)$
 \cite[Theorem 3]{LS}. Letting
 \begin{equation} \label{H(t)def}
  H(t) = H( \A  + \F(t), V)=  ({-i\nabla} - \A  - \F(t))^2  + V,
\end{equation}
where $\F(t)$ is as in \eq{Ftintro}, these time-dependent operators have the same domain: 
 $ \D(H(t))=\D(H)$ {for all $t \in \R$}, with $\D(A)$ the domain of the operator $A$. Bouclet, Germinet, Klein, and Schenker then used a well-known theorem due to Yosida  \cite [Theorem XIV.4.1]{Y} to solve the corresponding time-dependent Schr\" odinger equation by means of a unitary propagator, which  plays a major role in their analysis of the Liouville equation.
 
The  common domain condition does not hold for the time-dependent operators $H(t)$  in \eq{H(t)def}  under the more general conditions given in \eq{H(A,V)4}. But, as we shall see, we have a common quadratic form domain:  $Q(H(t)) = Q(H)$ for all $t$, where $Q(A)$ denotes the quadratic form domain of the self-adjoint operator  $A$. In view of 
\eq{Homegamma}, we may  take $H(t) \ge 1$ without loss of generality.  It then follows from  the closed graph theorem that
\begin {equation} \label{Gamma}
\Gamma (t,s) := H(t)^{\frac{1}{2}} H(s)^{-\frac{1}{2}} - I =  H(t)^{\frac{1}{2}}\left( H(s)^{-\frac{1}{2}} - H(t)^{-\frac{1}{2}}\right)
\end {equation}
is a family of bounded operators on $\H$ for all $t$ and $s$.  

We will prove an extension of Yosida's Theorem  if  the operators $\Gamma (t,s)$ meet certain conditions.  Before stating our theorem, recall  that a two-parameter family of unitary operators $U(t,s)$ on a Hilbert space $\H$, $s,t \in I \subset \R$   is called a unitary propagator if it satisfies
\begin{align}
\mathbf{1)}  \; \;\; \,  &U(t,r)U(r,s)  = U(t,s)  \label{U(t,r)U(r,s)},  \\
\mathbf{2)}  \; \;\; \,& U(t,t)  = I   \label{U(t,t)},  \\
\mathbf{3)}  \; \;\; \, & U(t,s) \text{ is jointly strongly continuous in  t  and  s.} \;\;\;\;\;\;\;\;\;\;\;\;\;\;\;\;\;\;\;\;\;\; \;\; \label{U(t,s)_st_conti}
\end{align}

 \begin{theorem} \label{WeakYosida} Let  $\I_0\subset \R$ be an open interval  and   $\H$  a Hilbert space.  Suppose that for each  $t\in \I_0$ we are given a self adjoint operator $H(t) \ge 1$ on  $\H$ such that its form domain $Q(H(t)) $ is independent of $t$: there exists a dense subspace $Q$ of $\H$ such that  $Q(H(t)) = Q$ for all $t\in \I_0$.
  Suppose  that for  some closed subinterval  $\I \subset \I_0$ the following holds  for all  $\varphi \in \H$:
 \begin{itemize} 
\item[a)]
 ${\frac{1}{t-s}}\Gamma(t,s) \varphi$ is uniformly bounded and strongly continuous for all  $t,s \,(t \neq s)$ in $\I $.
\item[b)]
The limits
\begin{align}
\Gamma_1(u) \varphi := {\textstyle \lim_{ t \uparrow u}{\frac{1}{u-t} } }\Gamma(t,u)\varphi  \quad \text{and} \quad 
\Gamma_2(u) \varphi  := {\textstyle \lim_{ s \uparrow u}{\frac{1}{u-s} } }\Gamma(u,s)\varphi  \label{Gamma_1&Gamma_2}
\end{align}
exist uniformly in $u \in \I$ with  $\Gamma_1(u) + \Gamma_2(u) = 0$.
\end{itemize}
Then there exists a unique unitary propagator $U(t,s)$, $t, s \in \I$, such that  
\begin{align}  \label{U_Qinv}
U(t,s)Q &= Q \ , \, \\
i \partial_t \la \varphi , U(t,s) \psi \ra &=   \la H(t)^{\frac{1}{2}}\,\varphi, H(t)^{\frac{1}{2}} U(t,s) \psi   \ra \;\;\; \mbox{for all $\varphi, \psi \in Q$}\, ,  \label{w_partial_t}\\
\label{w_partial_s}
 i \partial_s \la \varphi, U(t,s) \psi \ra &= -\,  \la H(s)^{\frac{1}{2}}\,U(t,s)^* \varphi,  H(s)^{\frac{1}{2}}\, \psi   \ra \;\;\;  \mbox{for all $ \varphi,  \psi \in Q$}\, .
\end{align}
\end{theorem} 
Related results for a time-dependent Hamiltonian with a time-independent form domain can be found in  \cite[Theorem 8.1]{Ky} and  \cite[Theorems II.23, II.27]{Si3}.

Theorem~\ref{WeakYosida} is exactly what we need in view of the following theorem.

\begin{theorem} 
\label{Q_invariant} 
Let $H(t)$ be  given in \eq{H(t)def} with $\A$ and $V$ as in \eq{H(A,V)4}, adjusted so $H(t)\ge 1$.  Then the 
 form domains $Q(H(t))$ are  independent of $t$:
\begin{align}
 Q({H(t)}) = Q := Q({H}) \quad \text{for all} \;\; t\in \R.	\label {Q_invar}
\end{align}
Moreover, $\Gamma(t,s)$, defined in (\ref{Gamma}), 
 satisfies hypotheses $a)$ and  $b)$ of Theorem~\ref{WeakYosida}. 
\end{theorem}

 To give a  precise meaning to the Liouville equation \eq{Liouvilleeq} we need to introduce
 normed spaces  of measurable
covariant operators, which we  briefly describe here. We refer the reader to
\cite[Section~3]{BGKS} for background, details, and justifications.  We assume the setting of \eq{H(A,V)omega}-\eq{covintro}; given $\H  =\mathrm{L}^2(\mathbb{R}^d)$, let $\Hc$ denote the subspace of  functions with compact support.  We set $\Ll= \Ll (\Hc,\H)$ to be the vector
space of linear operators on  $\H$ with domain  $\Hc$, and let $\K_{mc}$ be the vector space of measurable  covariant maps $A \colon \Omega \to \Ll$, $\omega \mapsto A_\omega$.  Throughout the article, we simplify the notation and write $A = A_\omega$. We also identify maps that agree $\P$-a.e. The map  $A_\omega$ is measurable if the functions $\omega \to \la \varphi,
A_{\omega} \varphi\ra$ are measurable for all $\varphi \in \H_{c}$, and $A_\omega$ is
covariant if  it satisfies \eqref{covintro}. 
$A_\omega$ is locally bounded if $ \|A_\omega \chi_x\|< \infty$ and $ \|\chi_x A_\omega \| <\infty$ for all $ x \in \Z^d$,
and we denote by $\K_{mc,lb}$ the subspace of locally bounded
operators.  If $A_\omega \in
\K_{mc,lb}$, then $\mathcal{D}(A_\omega^*) \supset \H_{c}$, and we set ${A_\omega}^\ddagger:={A_\omega}^*{\big|}_{\H_{c}}$. 

We  introduce  norms on $ \K_{mc,lb}$, given by 
\begin{equation}
 \tnorm{{A_\omega}}_i^{i} := \E\bigl\{ \tr \{\chi_0 |\overline{A_{\omega}}|^{i}
    \chi_0\}\bigr\} \quad \text{for} \,\, i=1,2\ , \quad  \tnorm{{A_\omega}}_\infty :=  \|\, \|{A_\omega} \| \,\|_{\mathrm {L}^{\infty} (\Omega, \P)},
\end{equation}
and define  the normed spaces
\begin{equation}
  \K_{i}: = \{ {A_\omega}\in \K_{mc,lb}; \, \tnorm{{A_\omega}}_i<\infty \},\ \quad
  i=1,2,\infty . 
\end{equation}
$\K_{\infty}$ is a Banach space,  $\K_1$ a normed space with closure $\overline{\K_1}$, and  $\K_{2}$ is a Hilbert
space with the inner product
\begin{equation}
   \la \la A_\omega, B_\omega\ra\ra  :=  \E \left \{ \tr \{ (A_{\omega} \chi_0
   ) ^{*}B_{\omega} \chi_0\} \right\}. \label{dblinnerprod_1}
\end{equation}
We consider the following linear functional on the space $\K_1$: for $A_\omega \in \K_1$, 
\begin{align} \label{tuvintro}
\T(A_\omega) := \E \left\{\tr\,  \chi_0 A_\omega \chi_0\right\}.
\end{align}
Since $|\T(A_\omega)| < \tnorm{A_\omega}_1$,  $\T$ is well-defined on $\K_1$. In addition, $\T$ is the trace per unit volume due to the Birkhoff ergodic theorem  \cite[Proposition 3.20]{BGKS}.

The spaces $\K_i$, $i=1,2$, are left
and right $\K_\infty$-modules. We define left and right multiplication for $B_\omega\in\K_\infty$ and ${A_\omega}\in\K_1$ or $\K_2$
by
\begin{equation} \label{leftrightM}
B_\omega \odot_L {A_\omega} = B_\omega  {A_\omega}\,,  \;\;\;
{A_\omega} \odot_R B_\omega = \left(B_\omega^\ast \odot_L {A_\omega}^\ddagger\right)^\ddagger =
{A_\omega}^{\ddagger\ast} B_\omega\,,
\end{equation}
and set
\begin{equation} \label{script_U}
{\U}(t,s)(A_\omega) := U_\omega(t,s)\odot_L A_\omega \odot_R U_\omega(s,t) \quad \mbox{for $A_\omega\in\K_i$}.
\end{equation}
 $\U(t,s)$ leaves $\K_i $ invariant  for $i= 1,2, \infty$ (See Proposition~\ref{propUomega}).

We will show that the Liouville equation \eqref{Liouvilleeq} can be  solved  in a weak sense  in the space $ \K_2$.  We let  $Q^{(0)}$ denote the subspace of $\K_2$ given by
  \begin{align}
Q^{(0)} &=  \left \{{A_\omega}\in\K_2 |\, H_\omega^{{\frac{1}{2}}} {A_\omega}, H_\omega^{{\frac{1}{2}}} {A_\omega^{\ddagger}}\in\K_2 \right \} =\left \{{A_\omega}\in\K_2 |\, H_\omega(t)^{{\frac{1}{2}}} {A_\omega}, H_\omega(t)^{{\frac{1}{2}}} {A_\omega^{\ddagger}}\in\K_2 \right \}  
\end{align} 
for all $t \in \R$.  (We refer the reader to Section 4.1 -- especially Propositions~\ref{DL0_Core} and \ref{HR=HLdagger} --  for the argument that $Q^{(0)}$  does not depend on $t$.) The quadratic forms $\widetilde \HH_L$, $\widetilde \HH_R$, and $\widetilde {\mathbb{L}}$  are defined   on  $Q^{(0)}$ as follows:
\begin{align}
	&\widetilde{\HH}_{L, t} (A_\omega,B_\omega) =   \la   \la H_\omega(t)^{\frac{1}{2}} A_\omega,H_\omega(t)^{\frac{1}{2}}B_\omega  \ra    \ra \ , \\
	&\widetilde{\HH}_{R, t} (A_\omega,B_\omega) =   \la  \la H_\omega(t)^{\frac{1}{2}} B_\omega^{\ddagger},H_\omega(t)^{\frac{1}{2}}A_\omega^{\ddagger} \ra \ra \ , \; \text{and} \\
	&\widetilde{\LL}_{t} =\widetilde{ \HH}_{L,t} - \widetilde{\HH}_{R,t}\ . \label{L_tilde}
\end{align} 
 
Recall  that ${\zeta}_\omega =  f(H_\omega)$ with $f$ real and bounded. We assumed    \eq{key_hypothesis}, which  implies $[x_{k},{\zeta}_\omega] \in \K_{2}$ for all $k=1,2,\ldots,d$, the condition used in \cite{BGKS}.
We set
\begin{align}
\label{Ptintro}
{\zeta}_\omega(t)&  = f(H_\omega(t)). 
\end{align}
\begin{theorem} \label{thmrho}  Let $H_{\omega}$ be the ergodic magnetic Schr\"odinger operator in \eq{H(A,V)omega}-\eq{covintro}, adjusted so $H_{\omega}\ge 1$, and let  $H_{\omega}(t)$ be  as in \eq{H(t)def}.  Let  
 ${\zeta}_\omega$ be as above satisfying \eqref{key_hypothesis}.  Then,
\begin{align}
\varrho_\omega(t) &: =\lim_{s \to -\infty}{ \U}(t,s)\left( {\zeta}_\omega
\right) = \lim_{s \to -\infty}{ \U}(t,s)\left( {\zeta}_\omega(s) \right) \\
&= {\zeta}_\omega(t) - i \int_{-\infty}^t \mathrm{d} r \,\mathrm{e}^{\eta {r_{\! -}}}{ \,\U}(t,r) \left([ \mathbf{E} \cdot \x, {\zeta}_\omega(r) ]\right)  \label{defrho3intro}
\end{align} 
is well defined in $ \K_2$.
Moreover,  $\rho_{\omega} (t) \in Q^{(0)}$,   and it is the unique solution of the Liouville equation in the following sense:  for all $A_{\omega} \in  Q^{(0)}$ we have
\begin{equation} 
\left\{ \begin{array}{l} i\partial_t \la\la A_\omega, \varrho_\omega(t)\ra\ra = \widetilde{\LL}_t (A_\omega, \varrho_\omega(t)) \\ 
\lim_{t \to -\infty} \la\la A_\omega, \varrho_\omega(t) \ra \ra = \la\la A_\omega , {\zeta}_\omega \ra \ra
\end{array}
\right. \, . \label {GenLiouvilleeq}
\end{equation}
\end{theorem}
We    will actually  prove a stronger   version of Theorem~\ref{thmrho} (cf.  Theorem~\ref{thm_rho}).

 %%%%%%%%%%%%%%%%%%%%%%%%%%%%%%%%%%%%%%%%%%%%
 %%%%%%%%%%%%%%%%%%%%%%%%%%%%%%%%%%%%%%%%%%%%

 \section{The extension of Yosida's Theorem}

 In this section we prove  Theorem~\ref{WeakYosida}.  We assume throughout   the section that   $H(t)$  and $\Gamma(t,s)$   satisfy the hypotheses of Theorem~\ref{WeakYosida}.

We define unitary operators $U_k(t,s)$, 
$k=1,2,\ldots$, and $s,t\in \I$, a closed subset of $\R$  as follows:
\begin{align} \label{U_k}
 U_k(t,s)  &=  \e^{ -i(t-s) H(m + {\textstyle \frac{j-1}{k}}) }   \;\;\; \mbox{for}\;\;
 m+ \textstyle{\frac{j-1}{k}} \le  s \le t \le m  + \textstyle{\frac{j}{k}}, \\
& \hskip 1.6 in   t,s \in \I ,  m \in \Z  ,  j=1,2,\ldots,k,\nonumber \\
U_k(t,r) &= U_k (t,s) U_k(s,r)\;\; \;\;\;\;\;  \mbox{for }\;  s \le r \le t \in \I \,, \label{U_k_2}\\
U_k(t,s) &= U_k (s,t)^*  \quad  \quad  \; \;\;\;\;\;\;\;\;\;  \mbox{for all }\; s, t \in \I.    \label{U_k_3}
\end{align}
 \begin{lemma} \label{W_k_bdd} For all  $k=1,2,\ldots$ and  all $s,t \in \I$,
 	\begin{align} \label {W_k(t,s)}
	W_k (t,s) :=   H(t)^{\frac {1}{2}}U_k(t,s)H(s)^{-\frac {1}{2}} 
		\end{align}
is well defined as a  bounded operator.
  In fact, 
\begin{align} \label {W_k_bound}
  \| W_k (t,s) \| \le \left ( 1 + \textstyle{\frac{M_I}{k}} \right)^2 \, e^{M_I |t-s|}, 
\end{align}
where 
\begin{align} \label {def:M_I}
   M_I := \sup_{t,s \in \I , t \ne s}|t-s|^{-1}  \left \| \,\Gamma(t,s) \right \|.
\end{align}
\end{lemma}

\begin{proof}
 
Proceeding as in
\cite [Theorem XIV.4.1]{Y}, 
we  can write $W_k(t,s)$, using \eqref{U_k}-\eqref{U_k_3}, as follows: 
\begin{align} \label{W_k}
 W_k(t,s)   &= \left (\Gamma\left(t,\textstyle {\frac {[kt]}{k}}\right) + I \right )\left \{U_k\left(t,s\right) + {W_k}^{(1)}(t,s) + {W_k}^{(2)}(t,s) + \cdots \right \}  \nonumber \\
&\quad \;\;\; \times \left (\Gamma \left(\textstyle {\frac {[ks]}{k}}, s \right) + I \right ), 
\end{align}
where for $ s \le t$
\begin{align} \label{W_k^(1)}
{W_k}^{(1)}(t,s) &= \sum_{ku = [ks] +1}^{[kt]} {U_k (t,u)\Gamma \left( u, u- \textstyle{\frac{1}{k}}\right)U_k(u,s)}, \\
{W_k}^{(m+1)}(t,s) &= \sum_{ku = [ks] + m+1}^{[kt]} {U_k (t,u)\Gamma  \left( u, u- \textstyle{\frac{1}{k}}\right) {W_k}^{(m)} (u,s) }, \label{W_k^(m+1)}
\end{align}
and for $ s \ge t$
\begin{align} \label{W_k^(1)(t,s)_s>t}
{W_k}^{(1)}(t,s) &= \sum_{ku = [kt]+1  }^{[ks] } {U_k (t,u)\Gamma  \left(  u- \textstyle{\frac{1}{k}}, u \right)U_k(u,s)}, \\
{W_k}^{(m+1)}(t,s) &=   \sum_{ku = [kt] +1}^{[ks] -m}           {U_k (t,u)\Gamma \left(  u- \textstyle{\frac{1}{k}}, u\right) {W_k}^{(m)} (u,s) }. \label{W_k^(m+1)(t,s)_s>t} 
\end{align}
In both cases, we see that $W_k^{(m)}$ is bounded for each $k$ with the following bound:
\begin{align} 
\| {W_k}^{(m)} (t,s) \varphi \|  \le \textstyle {\frac{(t-s)^m}{m ! }} M_I^m  \|\varphi \|, \quad m=1,2,\ldots, \label{W_k^(m)_bound}
\end{align}
so  \eqref{W_k_bound} follows from \eqref{W_k} and \eqref{W_k^(m)_bound}.
\end{proof}

%%%%%%%

\begin{lemma} \label{U_k_Convergence}  The unitary operators  $U_k(t,s)$, $t,s \in \I$,  converge strongly as $k \to \infty$  to a unitary propagator $U(t,s)$.   That is, 
     \begin{equation} \label{U(t,s)}
U(t,s)\psi = \lim_{k \to \infty} U_k(t,s)\psi   \quad \text{for}\;\; t,s \in \I
\end{equation}  
defines a unitary propagator, the convergence being uniform on $\psi \in \H$.   
In addition, for $ \;j = 1,2, \cdots$, 
\begin{equation} \label{st_conv_W_k^(j)}
{W_k}^{(j)}(t,s) \rightarrow  W^{(j)}(t,s) \quad\text{strongly for} \;\, t,s \in \I,
\end{equation}
 where $ W^{(j)}(t,s)$ are bounded operators given by
\begin{align}
W^{(1)}(t,s) &=   \int_s^t{ du \,U(t,u)   \Gamma_{2}(u)  U(u,s)},   \label {W^(1)}\\
W^{(j+1)}(t,s) &=   \int_s^t{ du \,U(t,u)  \Gamma_{2}(u)  W^{(j)}(u,s)}. \label{W^(j)}
\end{align}
Furthermore,   for all $t,s \in \I$, we have that  
\begin{align}
 W(t,s)     &=  \lim_{k \to \infty}{ W_k(t,s) } = U(t,s) + W^{(1)}(t,s) + W^{(2)}(t,s) +  \cdots   \label{W(t,s)_as_limit} \\
 &  = H(t)^{\frac {1}{2}}U(t,s)H(s)^{-\frac {1}{2}},  \label{W(t,s)} 
\end{align}  
the limits being in the strong operator topology, 
is a bounded operator,  weakly continuous in $t$ on $H(s)^{  \frac{1}{2}}Q$  for $s\le t$ and weakly continuous in $s $ on  $H(t)^{  \frac{1}{2}}Q$ for   $s \ge t$.  
\end{lemma}
\begin{proof} 
We  will prove the lemma for $s\le t$, the case $s \ge t$ being similar. We first  prove \eqref{U(t,s)}.     By construction (cf.  \eqref{U_k}-\eqref{U_k_2}),   we have  $ U_k(t,s)Q \subset  Q$. 
Since 
\begin{align}
U_k(t,s)  & = U_k \left(t,\textstyle {\frac {[kt]}{k}}\right)  U_k\left(\textstyle {\frac {[kt]}{k}}, s \right) = \exp \left (-i \left(t - \textstyle{\frac{[kt]}{k}}\right)  H \left(\textstyle{\frac{[kt]}{k}} \right) \right) U_k\left(\textstyle {\frac {[kt]}{k}}, s \right) , \label{U_k(t,s)psi} 
\end{align}
 it follows that $ U_k(t,s)$ is jointly strongly continuous in  $t,s \in \I$,   and $ U_k(t,s)$ is weakly differentiable in $t$ ($t \ne \textstyle {\frac{j}{k}}$) on $Q$ in the following sense: for $\varphi, \psi \in Q$, 
\begin{align} \label{w_partial_t_(U_k)}
i \partial_t {\la \varphi, U_k(t,s) \psi \ra } & = \left \la H \left(\textstyle {\frac {[kt]}{k}}\right)^{\frac{1}{2}}\varphi \, ,  H\left(\textstyle {\frac {[kt]}{k}}\right)^{ \frac{1}{2}}  U_k (t,s) \psi \right \ra.
\end{align}
Fixing  $s_0 \in \I$  and writing  $\varphi = H(s_0)^{-\frac{1}{2}}\widetilde {\varphi}$, $\psi = H(s_0)^{-\frac{1}{2}}\widetilde {\psi}$, with    $\widetilde{\varphi}$, $\widetilde {\psi} \in \H$, we get 
\begin{align}
& i \partial_t {\la \varphi, U_k(t,s) \psi \ra } = \left \la H \left(\textstyle {\frac {[kt]}{k}}\right)^{\frac{1}{2}}H(s_0)^{-\frac{1}{2}}\widetilde{\varphi} \, ,  H\left(\textstyle {\frac {[kt]}{k}}\right)^{ \frac{1}{2}}  U_k (t,s) H(s_0)^{-\frac{1}{2}} \widetilde {\psi} \right \ra \nonumber \\
& \quad = \left \la \left (\Gamma \left(\textstyle {\frac {[kt]}{k}}, s_0 \right)  + I \right ) \widetilde{\varphi} , \left (\Gamma \left(\textstyle {\frac {[kt]}{k}}\, , t \right)  + I \right ) W_k(t,s) \left (\Gamma \left(s, s_0 \right)  + I \right ) \widetilde {\psi} \right \ra.  \label{eq5}
\end{align} 
Hence,   by hypothesis $a)$  of Theorem~\ref{WeakYosida}
 and Lemma~\ref{W_k_bdd},  $ i \partial_t {\la \varphi, U_k(t,s) \psi \ra }$ is bounded and is (piece-wise) continuous in $t$ for  $t \ne \textstyle {\frac{j}{k}}$. Moreover, the same argument repeated 
 for the other variable  (i.e., for $s$) gives, with  $\varphi, \psi \in Q$, 
\begin{align} \label{w_partial_s_(U_k)}
i \partial_s {\la \varphi, U_k(t,s) \psi \ra } & = - \left \la H \left(\textstyle {\frac {[ks]}{k}}\right)^{\frac{1}{2}} U_k (t,s)^* \varphi \, ,  H\left(\textstyle {\frac {[ks]}{k}}\right)^{ \frac{1}{2}}  \psi \right \ra ,
\end{align}
and that  $ i \partial_s {\la \varphi, U_k(t,s) \psi \ra }$ is bounded and continuous in s for  $s \ne \textstyle {\frac{j}{k}}$.

Thus, one may easily compute that  
\begin{align}
&  \la \varphi, \left (U_k(t,s) - U_n(t,s)\right ) \psi \ra   
   =  \int_{s}^{t} {\textstyle {\frac{d}{dr}}  \left \la \varphi \,,  U_n(t,r)U_k(r,s)H(s_0)^{-\frac{1}{2}}\widetilde {\psi}  \right  \ra dr } \\   
&   = i\int_{s}^{t} { \left \la \Gamma\left(\textstyle {\frac {[kr]}{k}}, \textstyle {\frac {[nr]}{n}}\right) H \left(\textstyle {\frac {[nr]}{n}}\right)^{\frac{1}{2}}U_n (t,r)^* \varphi \, ,  \widetilde \Gamma \left(\textstyle {\frac {[kr]}{k}}, r \right)     W_k (r,s)  \widetilde\Gamma \left(s, s_0 \right)    \widetilde \psi \right \ra  dr }\nonumber \\         
& \;\;\;  -i \int_{s}^{t} { \left \la H \left(\textstyle {\frac {[nr]}{n}}\right)^{\frac{1}{2}} U_n (t,r)^* \varphi \, ,\Gamma\left(\textstyle {\frac {[kr]}{k}}, \textstyle {\frac {[nr]}{n}}\right)  \widetilde\Gamma\left(\textstyle {\frac {[nr]}{n}}, r \right)   W_k (r,s)  \widetilde\Gamma\left(s, s_0 \right)  \widetilde \psi \right \ra     dr },\nonumber 
\end{align} 
where  we used $\widetilde \Gamma := \Gamma + I$.
By \eqref {def:M_I}, we have
\begin{align}
&  \left \| \Gamma \left(\textstyle {\frac {[kr]}{k}}, \textstyle {\frac {[nr]}{n}}\right)    \right \|  \le  \left | \textstyle {\frac {[kr]}{k}} - \textstyle {\frac {[nr]}{n}}   \right |M_I,  
\end{align}
and by Lemma~\ref{W_k_bdd},
\begin{align}
  & \left \|   H \left(\textstyle {\frac {[nr]}{n}}\right)^{\frac{1}{2}} U_n (t,r)^* \varphi   \right \|  = \left \| \left (\Gamma\left(\textstyle {\frac {[nr]}{n}}, t \right) + I \right )  W_n (r,t)\left (\Gamma \left(r, s_0 \right) + I \right) \widetilde \varphi  \right \|  \nonumber \\
&  \qquad \qquad \le  \left ( M_I \left |\textstyle{{\frac {[nr]}{n}}} - t    \right | + I     \right ) \left ( 1+ \textstyle{\frac{M_I}{n}}   \right )^2 \, e^{M_I (t-r)}\left ( M_I |r - s_0| + I  \right).  
\end{align} 
Since $\Gamma \left(\textstyle {\frac {[kr]}{k}}, r \right),  \Gamma \left(s, s_0 \right)$, and $W_k (r,s)$ are all uniformly bounded independent of $r,s$, and $k$,  and  since
\begin{equation} 
\left \|  \left (U_k(t,s) - U_n(t,s)\right ) \psi    \right \| = \sup_{\varphi \in Q, \|\varphi\|=1} {\left |  \la \varphi, \left (U_k(t,s) - U_n(t,s)\right ) \psi \ra    \right |}, 
\end{equation}
we see that $U_k(t,s) \psi$ converges uniformly on $Q$ for $s \le t$,  $s, t \in \I$.   
Since $U_k(t,s)$ is uniformly bounded, the limit in \eqref  {U(t,s)} exists uniformly on $\H$ for $s \le t$. 
  It is a simple exercise to show that  \eqref{w_partial_t_(U_k)} and \eqref{w_partial_s_(U_k)} also hold for the case $s \ge t$. Thus,  \eqref{U(t,s)} holds for all $ t, s \in \I$, and we  conclude that $U(t,s)$, $t,s \in  \I$, is a unitary propagator.

We can prove \eqref{st_conv_W_k^(j)}  for $s \le t$  as follows.
Recalling \eqref{W_k^(1)}, one can easily see that, given $\varphi \in \H$ and letting  $k \to \infty$, we get 
\begin{align}
&\abs{ W_k^{(1)}(t,s) \varphi - \int_s^t{ du \,U(t,u)\Gamma_2 (u)U(u,s) \varphi}}   \\
& \;\; \le   \sup_{|u-v| \le \textstyle{ {\frac{1}{k} }}} \left\|  U_k (t,u) k\,
\Gamma \left( u, u- \textstyle{\frac{1}{k}}\right)U_k(u,s) \varphi -  U (t,u)\Gamma_2 (u) U(u,s) \varphi   \right\||t-s| \nonumber \\
& \qquad \; +  \sup_{|u-v| \le \textstyle{ {\frac{1}{k} }}} {\| U (t,u)\Gamma_2 (u) U(u,s) \varphi - U (t,v)\Gamma_2(v) U(v,s) \varphi  \| }\, |t-s| \;\; \to \; \; 0 ,\notag
\end{align}
using  \eqref{U(t,s)_st_conti}, \eqref{U(t,s)},  and the hypotheses  of Theorem~\ref{WeakYosida}.     (For $s\ge t$, we simply use $\Gamma_1 = -\Gamma_2$.) Hence, using induction and  \eqref{W_k^(m+1)}, one can show \eqref{st_conv_W_k^(j)} for all $j$ and for all $s, t \in \I$.  \eqref{W(t,s)_as_limit} now follows from \eqref{W_k} and \eqref{W_k^(m)_bound}.

It remains to prove that that $W(t,s)$ for $s\le t$  is given as in \eqref{W(t,s)} and that it is weakly continuous in $t$ on $H(s)^{  \frac{1}{2}}Q$.   
Note that for all $\varphi \in Q$, we have 
\begin{align}
   H(t)^{  \frac{1}{2}} U_k(t,s) \varphi 
    = W_k(t,s)H(s)^{  \frac{1}{2}}\varphi 
 \rightarrow W(t,s)H(s)^{  \frac{1}{2}}\varphi
\end{align}
Since  $H(t)^{  \frac{1}{2}} $  is closed, it follows that $ U(t,s)\varphi  \in \D  (H(t)^{  \frac{1}{2}} ) = Q$, and   for all $\varphi \in Q$,
\begin{align} \label{eq8}
  H(t)^{  \frac{1}{2}}U(t,s)\varphi = W(t,s)H(s)^{  \frac{1}{2}}\varphi.
\end{align}
  Letting $\varphi = H(s)^{ - \frac{1}{2}}\widetilde \varphi$ in \eqref{eq8}  with $\widetilde \varphi \in \H$,  \eqref{W(t,s)} now follows. By hypothesis $a)$ of Theorem~\ref{WeakYosida}, we also note that  $H(t)^{\frac {1}{2}} \varphi $   is continuous in $t \in \I$ for each  $\varphi \in Q$. To see this, if $t' \in \I$ and $\varphi  = H(t')^{-\frac{1}{2}} \psi \in Q$, then 
\begin{align} \label{st_conti_H1/2}
	\left(H(t)^{\frac {1}{2}} - H(t')^{\frac {1}{2}}  \right)\varphi = \left(H(t)^{\frac {1}{2}} - H(t')^{\frac {1}{2}}  \right) H(t')^{-\frac{1}{2}} \psi = \Gamma (t,t') \psi \to 0 
	\end{align}
as $t \to t'$. Thus, given  $ \varphi , \psi \in Q$, and setting  $ \widetilde\varphi = H(s)^{ \frac{1}{2}} \varphi$,  
 we see that 
\begin{align}
&\left \la \psi , \left (  W(t,s) - W(t' , s)   \right ) \widetilde \varphi \right \ra = \left \la \psi ,   (H(t)^{  \frac{1}{2}}U(t,s) - H(t')^{  \frac{1}{2}}U(t',s) )  \varphi \right  \ra  \\
&\quad = \left \la   (H(t)^{  \frac{1}{2}} - H(t')^{  \frac{1}{2}} ) \psi , U(t,s)   \varphi \right  \ra + \left \la H(t')^{  \frac{1}{2}}\psi , \left ( U(t,s) - U(t',s)   \right )  \varphi \right  \ra     \rightarrow 0  \nonumber
\end{align}
as  $t \to t'$ by \eqref{eq8}-\eqref{st_conti_H1/2} and the fact that 
$U(t,s)$ is a unitary propagator.  
\end{proof}

.
 
We now present a proof of Theorem \ref {WeakYosida}. In the proof, $[k]$ will denote the largest integer less than or equal to $k$.

 \begin{proof} [Proof of Theorem \ref {WeakYosida}.]
      Since $ i \partial_t {\la \varphi, U_k(t,s) \psi \ra }$ is bounded and (piece-wise) continuous in $t$ for  $t \ne \textstyle {\frac{j}{k}}$ by \eqref{eq5}, which holds for all $t,s \in \I$,   we have    
\begin{align}
& i \left[\la \varphi, U_k(t,s) \psi \ra - \la \varphi, U_k(r,s) \psi \ra \right ] =  \int_r^t {\left \la H \left(\textstyle {\frac {[kl]}{k}}\right)^{\frac{1}{2}} \varphi, H \left(\textstyle {\frac {[kl]}{k}}\right)^{\frac{1}{2}} U_k(l,s) \psi \right \ra dl }. \label{UkUk}
\end{align} 
Since $H(l)^{\frac{1}{2}}$ is strongly continuous in $l$ (by \eqref{st_conti_H1/2}), it follows from \eqref{W(t,s)_as_limit} and \eqref{W(t,s)} that the integrand in \eqref{UkUk} converges   as $k \to \infty$: for $\varphi$, $\psi \in Q$,
\begin{align}
&\left \la H \left(\textstyle {\frac {[kl]}{k}}\right)^{\frac{1}{2}} \varphi, H \left(\textstyle {\frac {[kl]}{k}}\right)^{\frac{1}{2}} U_k(l,s) \psi \right \ra  \rightarrow   \left \la H(l)^{\frac{1}{2}} \varphi,   W(l,s) H(s)^{ \frac{1}{2}}   \psi \right \ra.  \label {form5}  
\end{align}
Taking limits on both sides of \eqref {UkUk} yields
\begin{align}\label{form 6.5}
& \lim_{k \to \infty} i \left[\la \varphi, U_k(t,s) \psi \ra - \la \varphi, U_k(r,s) \psi \ra \right ]   \\
& \qquad =  \int_r^t {   \left \la H(l)^{\frac{1}{2}} \varphi,   W(l,s) H(s)^{\frac{1}{2}}  \psi \right \ra dl } =    \int_r^t {   \left \la H(l)^{\frac{1}{2}} \varphi,  H(l)^{\frac{1}{2}}  U(l,s)  \psi \right \ra dl }, \nonumber
\end{align}
 first equality being justified by  \eqref{form5} and  dominated convergence, and the last equality  by \eqref{eq8}. 
Since  $W(l,s)$ is weakly continuous in $l$ on $H(s)^{\frac{1}{2}}Q$ (by Lemma~\ref{U_k_Convergence}) and $H(l)^{\frac{1}{2}}$ is strongly continuous in $l$,  the integrals in \eqref{form 6.5} are well-defined. 
.Thus,
\begin{align}
 i \left[\la \varphi, U(t,s) \psi \ra - \la \varphi, U(r,s) \psi \ra \right ]  =    \int_r^t {  dl  \left \la H(l)^{\frac{1}{2}} \varphi,  H(l)^{\frac{1}{2}}  U(l,s)  \psi \right \ra } \label{eq10}
\end{align}
by \eqref{U(t,s)} of Lemma~\ref{U_k_Convergence}, and, therefore,
\begin{align}
i \partial_t \la \varphi , U(t,s) \psi \ra &=  \left \la H(t)^{\frac{1}{2}}\,\varphi, H(t)^{\frac{1}{2}} U(t,s) \psi \right \ra \;\; \mbox{for all $\varphi, \psi \in Q$}. \label{partial_t_U(t,s)}
\end{align}

A similar proof,   using \eqref{w_partial_s_(U_k)},  yields,   for all $t,s \in \I$ ,
\begin{align}
 i \partial_s \la \varphi, U(t,s) \psi \ra &= -\,\left \la H(s)^{\frac{1}{2}}\,U^*(t,s) \varphi,  H(s)^{\frac{1}{2}}\, \psi \right \ra \;\; \mbox{for all $ \varphi,  \psi \in Q$}.
\end{align}

We now show the uniqueness of the solution of \eqref{w_partial_t}-\eqref{w_partial_s}. Let $U(t,s)$ and $\widetilde U(t,s)$ be two propagators satisfying \eqref{w_partial_t}-\eqref{w_partial_s}. Then, for all $\varphi, \psi \in Q,$ 
\begin{align}
&{i \partial_t {\la \varphi, U(s,t)\widetilde U(t,s) \psi \ra }}  \\
& \; = 
\left \la H(s)^{\frac{1}{2}}\,U(t,s)^* \varphi,  H(s)^{\frac{1}{2}}\,\widetilde U(t,s) \psi \right \ra  - \left \la H(s)^{\frac{1}{2}}\,U(t,s)^* \varphi,  H(s)^{\frac{1}{2}}\,\widetilde U(t,s) \psi \right \ra =0.  
 \nonumber
\end{align}
Thus, $U(s,t) \widetilde U(t,s)$ is constant in $t$, and letting $t=s$, we see that $U(s,t) \widetilde U(t,s) = I$,
so $U(t,s) = \widetilde U(t,s)$.  
\end{proof}
 
\section{The time-dependent magnetic Schr\"odinger operators and the common quadratic form domain}

In this section,  we let $H(t)$ be  given in \eq{H(t)def}  with $\A$ and $V$ as in \eq{H(A,V)4}, adjusted so $H(t)\ge 1$, and  prove Theorem~\ref{Q_invariant}.

\begin{proof}[Proof of Theorem~\ref{Q_invariant}] 
In view of \eqref{form_bound} and \eqref{H(t)def},
\begin{align}
	 Q(H(\A,V)) &= Q(H(\A,V_+)),  \text { and} \label {eq15}\\ 
	 Q(H(\A + \F(t),V)) &= Q(H(\A + \F(t),V_+)) \label{eq16}.
\end{align}
Thus, to prove the first part of the theorem, namely \eqref{Q_invar}, it suffices to show
\begin{align}
Q(H(\A, V_+)) = Q(H(\A + \F(t), V_+)).
\end{align}
Let $q_{A}$ denote the quadratic form associated with the operator $A$. Given  $\varphi, \psi \in C_c^{\infty}(\R^d)$, we note that 
\begin{align}\notag
 q_{H(t)} (\varphi, \psi) &= \la (-i\nabla - \A  - \F(t))\varphi, (-i\nabla - \A  - \F(t))\psi \ra +   \la \varphi, V\psi \ra    \\
	&   = {q}_{H} (\varphi, \psi) + \S_t(\varphi, \psi), \label{eq20}
\end{align}
where  
\begin{align}
\S_t(\varphi, \psi) &= \la -\F(t)\varphi, (-i\nabla - \A )\psi \ra + \la (-i\nabla - \A )\varphi, -\F(t)\psi \ra  +    \la \F(t)\varphi,\F(t)\psi \ra  \nonumber \\
  & =  -2\, \mathrm {Re} \la \varphi,  \F(t)\cdot(-i\nabla - \A )\psi \ra + |\F(t)|^2 \la \varphi, \psi \ra.\label{eq22}
\end{align}
Since ${\S_t(\varphi, \psi)}$ is a symmetric quadratic form  and  \
\begin{align}
	|\S_t(\psi, \psi)| 
	&\le 2 ( \textstyle{\frac{1}{\delta}} \, \|\psi \|  )(\delta \, \|\F(t) \cdot (-i\nabla - \A)\psi  \|)    + |\F(t)|^2 \|\psi\|^2 \nonumber \\
	 &\le \delta^2 |\F(t)|^2\, \| (-i\nabla - \A)\psi  \|^2 +  (\textstyle{\frac{1}{\delta^2}} + |\F(t)|^2)  \, \|\psi \|^2 
\end{align}
by  the Cauchy-Schwarz inequality, we have, 
for suitable $\delta$ (e.g., take $\delta =  \textstyle{\frac{1}{\sqrt{2}|\F(t)|}}$ and set $\alpha := \delta^2 |\F(t)|^2 = \textstyle{\frac{1}{2}} < 1$ and $\beta := \textstyle{\frac{1}{\delta^2}} + |\F(t)|^2 = 3|\F(t)|^2$),  
\begin{align}
	|\S_t(\psi, \psi)| &\le \alpha \,q_{H(\A,0)} (\psi,\psi)  + \beta ||\psi ||^2 \le \alpha \,  q_{H(\A,V_+)} (\psi,\psi) + \beta ||\psi ||^2   
	\end{align}
for all $\psi  \in Q(H(\A,V_+))$ and  $0\le \alpha <1$.  It follows from \cite[Theorem~X.17]{RS2} that $ Q(H(\A + \F(t), V_+)) =  Q(H(\A, V_+))$, and this proves \eqref{Q_invar}. 

To finish, we prove that $\Gamma(t,s)$, given in \eqref{Gamma}, satisfies  hypotheses a) and b) of Theorem~\ref{WeakYosida}. Given  $\varphi, \psi \in Q$, it follows from \eqref{eq20} that
\begin{align}
	&q_{H(t)} (\varphi, \psi ) - q_{H(s)} ( \varphi, \psi ) =\left \la \varphi, \left (\F(t)^2 - \F(s)^2 \right ) \psi \right \ra -2 \left \la \varphi, \left ( \F(t) - \F(s) \right )\cdot \Db \psi \right \ra , 
\end{align}
where  $\Db = \Db(\A)$ is the closure of $(-i\nabla - \A)$ as an operator from $\mathrm{L}^2(\mathbb{R}^d)$ to $\mathrm{L}^2(\mathbb{R}^d ; \mathbb{C}^d)$. Thus, letting $\varphi = H(s)^{-\frac 1 2} \widetilde \varphi,  \psi = H(s)^{-\frac 1 2} \widetilde \psi$ with $ \widetilde \varphi , \widetilde \psi \in \H$, we have
\begin{align}\label{C(t,s)} 
	& C(t,s) : = H(s)^{-\frac{1}{2}}(H(t) - H(s))H(s)^{-\frac{1}{2}} \\
	       & \;\;\;\; = H(s)^{-\frac{1}{2}}(\F(t)^2- \F(s)^2) H(s)^{-\frac{1}{2}} -2  H(s)^{-\frac{1}{2}} \left\{ \F(t) - \F(s)    \right \}\cdot \Db 
 H(s)^{-\frac 1 2}. \nonumber
\end{align}
Since $\F(t) \in  C^1(\R; \R^d)$ and  
$  \left \| {\Db}  H^{-\frac 1 2}   \right \| \ \le
C_{\alpha,\beta}  
$  
with $C_{\alpha,\beta}$ a constant  \cite[Proposition 2.3 $(i)$]{BGKS}, we have that $C(t,s)$, $H(s)^{\frac{1}{2}}C(t,s) $, ${\textstyle \frac{1}{t-s}}C(t,s)  $, and ${\textstyle \frac{1}{t-s}} H(s)^{\frac{1}{2}}C(t,s) $ are all uniformly bounded in norm for $t,s$ $ (t \ne s) $ in $\I$. Let
\begin{align} 
	 &	\widetilde N_{\I}:= \sup_{t,s \in \I , t \ne s} \left \| (t-s)^{-1} H(s)^{\frac{1}{2}} C(t,s)   \right \|  < \infty \ . \label{Ntilde_I}
 \end{align}

In addition  we   have that  $s \mapsto H(s)^{-\frac{1}{2}} $ is continuous in norm.    Indeed, by the  so-called  Dunford-Taylor formula, we have
\begin{align} \label{H(t)^{1/2}Gamma}  
&   (H(s)^{-\frac{1}{2}} - H(t)^{-\frac{1}{2}} )   = 
 \textstyle {\frac{1}{\pi}} \int_{0}^{\infty} {{\lambda}^{-\frac{1}{2}}  \left( \left(H(s) + \lambda \right) ^{-1} - \left ( H(t)+ \lambda \right )^{-1} \right)   \, d\lambda} \\
&\qquad \qquad   = \textstyle {\frac {1}{\pi}} \int_{0}^{\infty} { {\lambda}^{-\frac{1}{2}}  \left(H(t) + \lambda \right)^{-1}  H(s)^{\frac{1}{2}} C(t,s)\, H(s)^{\frac{1}{2}} \left( H(s)+ \lambda \right)^{-1}   \, d\lambda}. \nonumber
\end{align}
 Since    $ H(s)^{\frac{1}{2}} C(t,s)   \to 0$ in norm (as $t \to s$) by \eqref{Ntilde_I}, and $ H(t)^{\frac{1}{2}} \left(H(t) + \lambda \right)^{-1} $   and $\left( H(s)+ \lambda \right)^{-1}$ are uniformly bounded for $t, s \in \I$ with
\begin{align}
	 H(t)^{\frac{1}{2}} \left(H(t) + \lambda \right)^{-1} < \frac{C_0}{2\sqrt{\lambda}} \, , \   \text { and }  \left( H(s)+ \lambda \right)^{-1}< \frac {1}{1+ \lambda} \ , \label{bounds} 
\end{align}
where $C_0$ being a constant, the claim now follows.

Now, \eqref{C(t,s)}, together with norm continuity of $H(s)^{-\frac{1}{2}}$ and $\F(t) \in  C^1(\R; \R^d)$,  also shows that
$C(t,s)$, $H(s)^{\frac{1}{2}}C(t,s)$, \, ${\textstyle \frac{1}{t-s}}C(t,s) $, and ${\textstyle \frac{1}{t-s}} H(s)^{\frac{1}{2}}C(t,s)  $ are all jointly continuous in norm for $t,s$ $ (t \ne s) $ in $\I$. Furthermore,  
	\begin{align} \label{C(s)}
	C(s)&:= \lim_{   t \uparrow s } {\textstyle {\frac {1} {(t-s)}} C (t,s)} \\
	& =   2 H(s)^{-\frac{1}{2}} \F'(s) \cdot \F(s) H(s)^{-\frac{1}{2}} - 2H(s)^{-\frac{1}{2}} \F'(s) \cdot \Db H(s)^{-\frac{1}{2}}  \notag 
	\end{align}
exists boundedly (in norm).

With \eqref{Ntilde_I} and \eqref{bounds}, we now have 
\begin{align}
&  \left \| \textstyle {\frac{1}{\pi}}\int_{0}^{\infty}  {\lambda}^{-\frac{1}{2}}    H(t)^{\frac{1}{2}} \left(H(t) + \lambda \right) ^{-1}   H(s)^{\frac{1}{2}}  C(t,s) H(s)^{\frac{1}{2}} \left( H(s)+ \lambda \right)^{-1} \varphi   \, d\lambda \,\right \|  \nonumber \\
& \qquad  \qquad    \le \textstyle {\frac{1}{\pi}} \int_{0}^{\infty} { {\lambda}^{-\frac{1}{2}} \left(\frac{C_0 }{2\sqrt{\lambda}}\right)^2 \widetilde N_{\I}|t-s|  \|  \varphi \| \, d\lambda} \,  = \, C_{\I}|t-s| \| \varphi \|, \label{C_I}
\end{align}
where 
$C_{\I}$ = $\textstyle {\frac{C_0^2}{4\pi}  } \widetilde N_{\I} \int_{0}^{\infty} {  {\lambda}^{-\frac{1}{2}} \textstyle {\frac { d\lambda }{  \lambda}}} <  \infty$. Therefore, from \eqref{Gamma}, \eqref{H(t)^{1/2}Gamma}, and \eqref{C_I},
\begin{align} \label {Gamma_ito_C}
&\Gamma(t,s)\varphi = H(t)^{\frac{1}{2}} ( H(s)^{-\frac{1}{2}} - H(t)^{-\frac{1}{2}} ) \varphi    \\  
& \quad =   \textstyle {\frac{1}{\pi}}\int_{0}^{\infty} { {\lambda}^{-\frac{1}{2}}  H(t)^{\frac{1}{2}} \left(H(t) + \lambda \right) ^{-1} H(s)^{\frac{1}{2}} C(t,s)\, H(s)^{\frac{1}{2}} \left( H(s)+ \lambda \right)^{-1} \varphi   \, d\lambda},\nonumber
\end{align}
and this shows that, for each $\varphi \in \H$,
\begin{equation}  \label {Gamma_bdd}
\| \textstyle{\frac{1}{t-s}}\Gamma(t,s)  \varphi \| \le  C_{\I}  \| \varphi \| \quad \mbox {for all }\,  t,s \in \I, \,(t \ne s).
\end{equation}

To show strong continuity of $ \textstyle{\frac{1}{t-s}}\Gamma(t,s)$,  let us  set $T_{\lambda}(t) := H(t)^{\frac{1}{2}} \left (H(t) +\lambda \right)^{-1}$.  By \eqref{bounds}, $T_{\lambda}(t)$ is uniformly bounded in $t$ for all  $\lambda > 0$.  
To see that $T_{\lambda}(t)$ is strongly continuous on $\H$, let $\varphi \in \H$ and note that
\begin{align}
& \left (T_{\lambda}(t)-T_{\lambda}(t') \right )  \varphi  = \left \{ H(t)^{\frac{1}{2}} \left (H(t) +\lambda \right)^{-1} - H(t')^{\frac{1}{2}} \left (H(t') +\lambda \right)^{-1} \right \}\varphi  \nonumber \\ 
&\quad  = \left \{ \left( H(t)^{\frac{1}{2}} - H(t')^{\frac{1}{2}} \right ) H(t')^{-\frac{1}{2}}\right\}H(t')^{\frac{1}{2}} \left (H(t) +\lambda \right)^{-1} \varphi \label{Tconti_1}\\
&\qquad \quad  \; + \ \left\{H(t')^{\frac{1}{2}}  \left (H(t') +\lambda \right)^{-1} H(t')^{\frac{1}{2}}\right \} C (t', t) \, H(t')^{\frac{1}{2}} \left (H(t) +\lambda \right)^{-1}  \varphi \label{Tconti_2} \\
&\quad \rightarrow 0 \,\, \text{ as}\,\, t \to t'\,, \label{Tconti}
\end{align}
where \eqref{Tconti_1} goes to 0 by \eqref{def:M_I} (applied twice) and \eqref{bounds}, and 
\eqref{Tconti_2} goes to 0 by \eqref{C(t,s)}, \eqref{bounds}, and the fact that $F(t) \in C^1(\R; \R^d)$.

Thus, from \eqref{Gamma_ito_C},  
\begin{align}
	& \textstyle{\frac{1}{t-s} }\Gamma (t,s) \varphi =   \textstyle {\frac{1}{\pi}}\int_{0}^{\infty} { {\lambda}^{-\frac{1}{2}}  T_{\lambda}(t)   \left \{\frac{1}{t-s} H(s)^{\frac{1}{2}}C(t,s) \right\}\, T_{\lambda}(s) \varphi   \, d\lambda},  \label{1/(t-s)Gamma_1}
	\end{align}
and since ${\textstyle \frac{1}{t-s}}H(s)^{\frac{1}{2}}C(t,s) $ is jointly strongly continuous for $t, s \, (t \ne s)$ in $\I$, \eqref{Tconti} concludes that $\textstyle{\frac{1}{t-s}}\Gamma(t,s)  \varphi$ is also jointly strongly continuous for $t, s \, (t \ne s)$ in $\I$. This shows hypothesis $a)$ of Theorem~\ref{WeakYosida}.

It now follows from  \eqref{C(s)},
\eqref{Gamma_bdd}, 
 \eqref{Tconti}, and \eqref{1/(t-s)Gamma_1}  that
\begin{align}  \label{equality_4}
 \Gamma_1(u) 
 &:= \lim_{ t \uparrow u}  \textstyle{\frac{-1}{t-u}} \Gamma(t,u) \varphi    \\ 
&  = - \textstyle{\frac{1}{\pi}} \int_{0}^{\infty} {  {\lambda}^{-\frac{1}{2}}  H(u)^{\frac{1}{2}} \left(H(u) + \lambda \right) ^{-1} H(u)^{\frac 1 2 } C(u)\, H(u)^{\frac{1}{2}} \left( H(u)+ \lambda \right)^{-1}   \varphi   \, d\lambda},  \nonumber 
\end{align}
the limit being  uniform in $u\in \I$. Similarly, it   follows from
 \eqref{C(t,s)}, as in   \eqref{C(s)}, that
\begin{align}
\widehat{C}(u):= \lim_{ s \uparrow u}  \textstyle{\frac{ 1}{ u - s}} C(u,s)    
\end{align}
exists boundedly (in norm), and that $\widehat{C}(u)=C(u)$.
Therefore, it follows, as before,  
\begin{align}  \label{equality_5}
 \Gamma_2(u) 
 &:= \lim_{ s \uparrow u}  \textstyle{\frac{1}{u-s}} \Gamma(u,s) \varphi   = -\Gamma_1(u),  \end{align}
with the limit being   uniform in $u\in \I$.
This finishes the  proof of the theorem. 
 \end{proof}

%%%%%%%%%%%%%%%%%%%%%%%%%%%%%%%%%%%%%%%%%%%%%%%%%%%%%%%%%%%%%%%%%%%%%%%%%%%%%%%%%
%%%%%%%%%%%%%%%%%%%%%%%%%%%%%%%%%%%%%%%%%%%%%%%%%%%%%%%%%%%%%%%%%%%%%%%%%%%%%%%%%
%%%%%%%%%%%%%%%%%%%%%%%%%%%%%%%%%%%%%%%%%%%%%%%%%%%%%%%%%%%%%%%%%%%%%%%%% 

\section{ Quadratic forms on the Hilbert space $\K_{2}$}
In this section we define quadratic forms  on $\K_2$ in order to   give a precise meaning to the  Liouville equation~\eqref{GenLiouvilleeq}. 
Let $H_{\omega}$ and $H_{\omega}(t)$ be as in  Theorem~\ref{thmrho}.
For $\P$-a.e.\ $\omega$ we let $U_\omega(t,s)$ be the corresponding  unitary propagator given in
Theorem~\ref{WeakYosida}. 
As discussed in Section~\ref{secintro}, the spaces $\K_i, i= 1, 2,$ are left and right $\K_\infty$-modules, with  left and right multiplications  defined as in \eqref{leftrightM}.   
We  state  \cite[Prop. 4.7]{BGKS} without its proof. 

\begin{proposition} \label{propUomega}
For each $i=1,2,\infty$,  let
\begin{equation}  \label{UtsUst}
{\U}(t,s) (A_\omega) =U_\omega(t,s) \odot_LA_\omega \odot_R U_\omega(s,t)
\;\;\;\mbox{for \ $A_\omega \in \K_{i}\, $.}
\end{equation}
Then $\, {\U}(t,s)$ is a linear operator on $\K_{i}$,  $i=1,2,\infty$,  with
\begin{align}
{\U}(t,r)\, {\U}(r,s) &= {\U}(t,s) \, ,\\
{\U}(t,t) &=  I \,,\\
\left\{\,{\U}(t,s) (A_\omega)\right\}^\ddagger&=
 {\U}(t,s) (A_\omega^\ddagger)\, \label{UAUdag} .
\end{align}
Moreover, ${\U}(t,s)$ is unitary on $\K_2$ and an isometry in
 $\K_1$ and $\K_\infty$; it
 extends to an isometry
on $\overline{\K}_1$ with the same properties. In addition,
 ${\U}(t,s)$ is jointly strongly continuous
in $t$ and $s$ on  $\overline{\K}_1$ and $\K_2$.
\end{proposition}

\subsection{The operators $\H_L$ and $\H_R$ and their domains}
Since $U_\omega(t,s)$ depends on the electric field $\El$, let
 \begin{equation} \label{U0}
U_\omega(\El= 0, t,s)  = U^{(0)}_\omega(t-s) := \e^{-i(t-s) H_\omega} \, .
\end{equation}
For $\El=0$, we consider 
\begin{align}  \label{U00}
\U^{(0)}(t) (A_\omega) := U^{(0)}_\omega(t) \odot_L A_\omega \odot_R
 U^{(0)}_\omega(-t)
\;\;\;\mbox{for \ $A_\omega \in \K_{\odot}\, $.}
\end{align}
It turns out $U^{(0)}_\omega(t) \odot_L A_\omega $ is a strongly continuous semi-group on $\K_2$; there is a self-adjoint operator $\H_L$ on $\K_2$ such that
\begin{align}
	e^{- it \H_L}A_\omega = U^{(0)}_\omega(t) \odot_L A_\omega  = e^{- it H_\omega} \odot_L A_\omega  \label{U0A}
\end{align}
with a domain $\D(\H_L)$. Similarly, we  define ${\H_L(s)}$ by
\begin{align}
	 \e^{-it {\H_L(s)} }A_\omega	= \e^{-it {H_\omega(s)}  } \odot_L A_\omega .\label{HL(s)} 
\end{align}
Note that $\H_L, \H_R,  \H_L(s),$ and $\H_R(s)\ge 1$ since $H_{\omega} \ge 1$.

Let  $\G(t)$  be the strongly continuous  unitary group on $\K_2$ given by (cf. \cite[Lemma 4.13]{BGKS})
\begin{align}
	\G(t) (A_\omega) := G(t) \, A_\omega G(t)^*, \label{Script_G(t)} \;\; \text{where $G(t) = \mathrm{e}^{i\F(t)\cdot x}$;}
\end{align}
  if $A_\omega \in \K_2$ with $[x_j, A_\omega] \in \K_2$ for $j=1,2,\cdots, d$, then
\begin{align}
	\partial_t \G(t)(A_\omega)  = i  [\El(t)\cdot x, \G(t)(A_\omega)] = i \G(t) ([\El(t)\cdot x, A_\omega]). \label{partial_G(t)}
\end{align}
Since  
\begin{align}
	\e^{-it {H_\omega(s)}  } = \G(s)  \left(\e^{-it {H_\omega}  } \right) =  G(s) \, \e^{-it {H_\omega}  }G(s) ^*,
\end{align}%
it follows that 
\begin{align}\label{GHG}
 {H_\omega(s)}     =  G(s) {H_\omega}  G(s) ^*,
\end{align}
and 
\begin{align}\notag
	 \G(s)  \left( \e^{-it {\H_L}  } A_\omega \right)& = \G(s)   \left( \e^{-it {H_\omega} } \odot_L A_\omega \right) =     \G(s)   \left(  \e^{-it { H_{\omega}}  } \right)   \odot_L   \G(s)   \left(A_\omega \right)      \\
  &   =  \e^{-it {H_\omega (s) }  }  \odot_L   \G(s)  \left(A_\omega \right )   =  \e^{-it {\H_L (s) }  }   \G(s)  \left(A_\omega \right ).
\end{align}
This shows
\begin{align}
	\e^{-it {\H_L (s) }  } A_\omega =  \G(s) \left( \e^{-it {\H_L}  } \G(s)^* ( A_\omega)\right),
\end{align}
and thus
\begin{align}
	\H_L (s)  = \G(s) \,  \H_L  \, \G(s)^* . \label{HL(s)=G(s)HL}
\end{align}

Recall that  $ H_\omega(t)^{\frac{1}{2}} A_\omega \in \K_2$ if  $A_\omega \H_c \subset \D  (H_\omega(t)^{\frac{1}{2}} )$ and  $H_\omega(t)^{\frac{1}{2}} A_\omega : \H_c \to \H$ is in $\K_2$.
Let
\begin{align} \label{DL0}
  	\D_L^{(0)}  := \left\{ A_\omega \in \K_2 |  H_\omega^{\frac{1}{2}} A_\omega  \in \K_2 \right \}, \quad 	\D_{L,t}^{(0)}  := \left\{ A_\omega \in \K_2 |  {H_\omega(t)}^{\frac{1}{2}} A_\omega \in \K_2 \right \}.
\end{align}

\begin{proposition} \label{DL0_Core} $\D_{L }^{(0)}  $  and $\D_{L,t}^{(0)}  $ are operator cores for ${\H_L }^{\frac{1}{2}}$ and ${\H_L(t)}^{\frac{1}{2}}$, respectively, for all $t\in \R$.
Moreover, 
 $  \D_{L,t}^{(0)} = \D_L^{(0)}$ and $\D  ({\H_L(t)}^{\frac{1}{2}} ) = \D ({\H_L}^{\frac{1}{2}} )  $  for all $t\in \R$.   
\end{proposition}

\begin{proof}

We first show that $\D_{L }^{(0)} \subset \D   ({\H_L  }^{\frac{1}{2}})$ and  
\begin{align}
{\H_L}^{\frac{1}{2}}  A_\omega  = {H_\omega}^{\frac{1}{2}}A_\omega \quad \text{for all $A_\omega  \in \D_{L }^{(0)}$}.  \label{HLAHA}
\end{align}
Indeed, if $\varphi \in \H_c$\, and $A_\omega \in 	\D_L^{(0)}$, then as $t \to 0$
\begin{align}
	\frac{i}{t} \left ( e^{-it {H_\omega}^{\frac{1}{2}}} - 1 \right )  A_\omega \varphi \rightarrow {H_\omega}^{\frac{1}{2}} A_\omega \varphi \quad  
\end{align}
 for all $\varphi \in \H_c$\,. Thus, on account of \eqref{U0A}
\begin{align}
 {\H_L}^{\frac{1}{2}}  A_\omega = i \partial_t   \left. e^{-it {\H_L}^{\frac{1}{2}}}  A_\omega \right |_{t=0} = \lim_{t \to 0} \frac{i}{t} \left ( e^{-it {H_\omega}^{\frac{1}{2}}} - 1 \right )  A_\omega = {H_\omega}^{\frac{1}{2}}A_\omega.  
\end{align}    
Similarly,  $\D_{L,t}^{(0)} \subset \D  ({\H_L (t)}^{\frac{1}{2}})$ and 
\begin{align}
     {\H_L (t)}^{\frac{1}{2}} A_\omega = {H_\omega (t)}^{\frac{1}{2}}A_\omega  \quad \text{for all $A_\omega  \in \D_{L, t }^{(0)}$}.  \label{HLtAHtA}
       \end{align}
 
 We now show that $\D_{L,t }^{(0)}$ is dense in  $\D   ({\H_L (t)}^{\frac{1}{2}})$; that $\D_{L }^{(0)}$ is dense in  $\D   ({\H_L  }^{\frac{1}{2}})$ can be shown similarly. Let $A_\omega \in \D  ({\H_L (t)}^{\frac{1}{2}}) $ and $B_{ n, \omega}  := f_n \left({H_\omega (t)}^{\frac{1}{2}}\right) \odot_L A_\omega$, where $\{f_n\}$ is a sequence of smooth, measurable, and compactly supported functions that converges to $\delta$, the delta function. Then, $ B_{ n, \omega}\in \D_{L,t}^{(0)}$ and $B_{n, \omega} \to A_\omega $ in $\K_2$ as $n \to \infty$. Moreover,
\begin{align}
	&{\H_L(t)}^{\frac{1}{2}} B_{n, \omega} =  f_n  \left({ H_{\omega}(t)}^{\frac{1}{2}} \right) \odot_L  {\H_L(t)}^{\frac{1}{2}} A_\omega  \rightarrow   \, {\H_L(t)}^{\frac{1}{2}} A_\omega.
\end{align}

To conclude that  $\D_{L, t}^{(0)}$ is a core for ${\H_L (t)}^{\frac{1}{2}}$, it   suffices to show  that (see \cite[Theorem VIII.11]{RS1})
\begin{align}
e^{is \, {\H_L (t)}^{\frac{1}{2}}} \D_{L, t }^{(0)} \, \subset \, \D_{L, t }^{(0)} \quad \text {for all $s\in \R$}\label{eHLD_invar}.
\end{align}

Since
\begin{align}
	e^{is \, {H_\omega (t)}^{\frac{1}{2}}} \D  ({H_\omega (t) }^{\frac{1}{2}} ) = \D  ({H_\omega (t) }^{\frac{1}{2}}) \,, \label{eD=D}
\end{align}
given $B_\omega \in \D_{L, t }^{(0)}$, we note that  
\begin{align}
	  e^{is \, {\H_L (t)}^{\frac{1}{2}}} B_\omega \H_c =   \left (e^{is \, {H_\omega (t)}^{\frac{1}{2}}} \odot_L B_\omega  \right )\H_c = e^{is \, {H_\omega (t)}^{\frac{1}{2}}}  B_\omega \H_c   \subset \D  ({H_\omega (t) }^{\frac{1}{2}}). \label{eHLBinD}
\end{align}
Moreover, we have  
\begin{align}
&  {H_\omega  (t)}^{\frac{1}{2}}    e^{is \, {\H_L (t)}^{\frac{1}{2}}}  B_\omega  = {H_\omega (t)}^{\frac{1}{2}}   e^{is \, {H_\omega (t)}^{\frac{1}{2}}}   B_\omega   =  e^{is \, {\H_L (t)}^{\frac{1}{2}}}  {H_\omega  (t)}^{\frac{1}{2}}  B_\omega   \in \K_{2 .} \label{H(t)1/2EHLBinK2}
\end{align}
The desired \eqref{eHLD_invar} follows   from \eqref{eHLBinD} and  \eqref{H(t)1/2EHLBinK2}. Similarly,  $\D_{L }^{(0)}$ is a core for ${\H_L  }^{\frac{1}{2}}$.

If $A_{\omega} \in \D_L^{(0)}$,  
$H_\omega (t) ^{\frac{1}{2}} A_\omega$ is well-defined on $\H_c$ by Theorem~\ref{Q_invariant}.
Moreover, in view of \eqref{dblinnerprod_1}, \eqref{eq20}-\eqref{eq22},  
\begin{align}
 &   \la  \la H_\omega (t) ^{\frac{1}{2}} A_\omega, H_\omega (t) ^{\frac{1}{2}} A_\omega   \ra   \ra   \label{innerprod_HLt1/2}\\  
& 	\quad =      \la  \la  {H_\omega}^{\frac{1}{2}} A_\omega, {H_\omega}^{\frac{1}{2}} A_\omega   \ra  \ra - 2 \text{Re}\left \la \left \la A_\omega, \F(t)\cdot \Db\, A_\omega  \right \ra \right \ra +  |\F(t)|^2 \left \la \left \la A_\omega, A_\omega   \right \ra \right \ra . \nonumber
\end{align}
Since  $A_{\omega} \in \D_L^{(0)}$,  we have $|\F(t)|^2 \left \la \left \la A_\omega, A_\omega   \right \ra \right \ra =|\F(t)|^2 \tnorm{A_\omega}_2^2 <  \infty$,   
\begin{align}
	  \la \la  {H_\omega}^{\frac{1}{2}} A_\omega, {H_\omega}^{\frac{1}{2}} A_\omega  \ra\ra = \tnorm{{H_\omega}^{\frac{1}{2}} A_\omega}_2^2 < \infty\, ,
\end{align}
and by the Cauchy-Schwarz inequality
\begin{align}
	\left \la \left \la A_\omega, \F(t)\cdot \Db\, A_\omega  \right \ra \right \ra \le |\F(t)| \tnorm{ A_\omega}_2 \tnorm{{H_\omega}^{\frac{1}{2}}  A_\omega}_2 < \infty. \label{FromCauchySchwarz}
\end{align}
This shows  ${H_\omega(t)}^{\frac{1}{2}} A_\omega \in \K_2$  and, hence,
 $ \D_L^{(0)} \subset  \D_{L,t}^{(0)} $. A similar argument shows that $\D_{L,t}^{(0)} \subset  \D_{L}^{(0)}$, and we have $ \D_{L,t}^{(0)} = \D_{L}^{(0)}$.

To prove the last claim of the proposition,   by \eqref{HLAHA} and \eqref{HLtAHtA} we can use the arguments in \eqref{innerprod_HLt1/2}-\eqref{FromCauchySchwarz} to show that, given a Cauchy sequence $A^{(n)}_{\omega}$ in  $\K_{2}$  with $A^{(n)}_{\omega} \in  \D_{L}^{(0)} = \D_{L,t}^{(0)}$ for all $n$, 
 ${\H_L}^{\frac{1}{2}}A^{(n)}_{\omega} $ is a Cauchy sequence in $\K_{2}$ if and only if ${\H_L(t)}^{\frac{1}{2}}A^{(n)}_{\omega} $ is a Cauchy sequence in $\K_{2}$.  Since $ \D_{L}^{(0)}$ is an operator core for both ${\H_L}^{\frac{1}{2}}$ and ${\H_L(t)}^{\frac{1}{2}}$, we conclude that  $\D ({\H_L(t)}^{\frac{1}{2}}) = \D  ({\H_L}^{\frac{1}{2}})$.
 \end{proof}

In a similar fashion,  we define
$e^{- it \H_R} A_\omega  :=  A_\omega  \odot_R e^{it H_\omega}$ and $ e^{- it \H_R (s)} A_\omega  :=  A_\omega  \odot_R  e^{it H_\omega(s)}$,
where $\H_R \ge 1$ and $\H_R(t) \ge 1$ are self-adjoint operators on $\D(\H_R)$ and $\D(\H_R(t))$, respectively, and  let
\begin{align} \label{DR0}
	\D_R^{(0)}  := \left\{ A_\omega \in \K_2 |  H_\omega^{\frac{1}{2}} A_\omega^{\ddagger} \in \K_2 \right \}, \quad 	\D_{R,t}^{(0)}  := \left\{ A_\omega \in \K_2 |  {H_\omega(t)}^{\frac{1}{2}} A_\omega^{\ddagger} \in \K_2 \right \}.
\end{align}

Recalling that  $ \cJ : A \mapsto A^{\ddagger}$  is an anti-unitary map on $\K_2$, and that \eq{leftrightM} can be rewritten as  ${A_\omega} \odot_R B_\omega =\cJ  \left(B_\omega^\ast \odot_L {\cJ A_\omega}\right)$, we immediately have the following proposition.

\begin{proposition} \label{HR=HLdagger} We have $\cJ \, \H_L   \,  \cJ = \H_R$, that is, 
 $\D  ( {\H_R}^{\frac{1}{2}} ) = \left \{ \D  ( {\H_L}^{\frac{1}{2}} )\right \} ^{\ddagger}$ and  \begin{align}
	{\H_R}^{\frac{1}{2}} A_\omega = \left( {\H_L }^{\frac{1}{2}} A_\omega^{\ddagger} \right)^{\ddagger}   \quad \text{for all $ A_\omega \in \D   ( {\H_R}^{\frac{1}{2}} )$}.
\end{align}
Similarly, 	$\cJ \, \H_L (t) \,  \cJ = \H_R(t)$ for all $t$.
In particular, the appropriate modification 
of  Proposition ~\ref{DL0_Core} holds for ${\H_R}^{\frac{1}{2}}$ and  ${\H_R(t)}^{\frac{1}{2}}$.
\end{proposition}

\subsection{The quadratic forms $\HH_{L, t}$\,, $\HH_{R, t}$\,, and $\LL_t$}

Setting $Q(\H_L(t)) := \D  ({\H_L(t)} ^{\frac{1}{2}} )$ and  $Q(\H_R(t)) := \D  ({\H_R(t) }^{\frac{1}{2}})$, 
we define the following quadratic forms:
\begin{align} \label{quadform_HL}
	 	\HH_{L, t} (A_\omega,B_\omega)&  =  \la   \la {\H_L(t)}^{\frac{1}{2}} A_\omega,  {\H_L(t)}^{\frac{1}{2}} B_\omega   \ra   \ra \;\;\;\; \text{on }\ Q(\H_L(t)) \,  , \\
	  \HH_{R, t} (A_\omega,B_\omega)  & =   \la   \la {\H_R(t)}^{\frac{1}{2}} A_\omega,  {\H_R(t)}^{\frac{1}{2}} B_\omega  \ra  \ra  \;\;\;\;\text{on } \ Q(\H_R(t)) \,  ,  \label{quadform_HR}    
	 \end{align}
	 where the inner product $\la\la \cdot, \cdot \ra\ra$ on $\K_2$ is as in \eqref{dblinnerprod_1}.
By Proposition~\ref{DL0_Core}, 
\begin{align}
	Q(|\LL|) := Q(|\LL_t|) = Q(\H_L(t)) \cap Q(\H_R(t))\, , \label{QL_t} 
\end{align}
and  set
\begin{align}
	 |\LL_t| (A_\omega,B_\omega) & = \HH_{L,t}(A_\omega,B_\omega) + \HH_{R,t}(A_\omega,B_\omega) \;\;\;\text{on \,} Q(|\LL|) \, . \label{quadform_L} 
\end{align}
We remark  that   \eqref{quadform_HL}, \eqref{quadform_HR}, and  \eqref{quadform_L} are all closed forms, and on  $Q(|\LL|)$,  
\begin{align}
	 \HH_{R, t} (A_\omega,B_\omega)  = 	\HH_{L, t} \left(B_\omega^{\ddagger},A_{\omega}^{\ddagger}\right)
\end{align}
by Proposition~\ref{HR=HLdagger} and the fact that the map $\cJ\colon  A_{\omega} \mapsto A_\omega^{\ddagger} $ is anti-unitary on $\K_2$. 

There is  a  continuous bilinear map $\diamond : \K_2 \times \K_2 \to \K_1$, defined by  
\begin{align}
	\diamond (A_\omega, B_\omega) := A_\omega \diamond B_\omega = A_\omega B_\omega \quad \text{for}\quad   A,B \in \K_2\cap  \K_\infty  \ .
\end{align}
Given  $A,B \in \K_2$ and $C \in \K_{\infty}$, it is shown in \cite[Lemma 3.21]{BGKS} that we have  % ($\T$, the  trace per unit volume, is defined in \eqref{tuvintro})   
\begin{gather}\label{centralK2}
\T(A_\omega\diamond B_\omega) =
\left\la\left\la A_\omega^\ddagger, B_\omega \right\ra\right\ra  \, ,\\
 \label{central}
\T(A_\omega\diamond B_\omega) =
\T(B_\omega\diamond A_\omega)   \, , \\
\label{central2}
\T( (C_\omega\odot_L A_\omega)\diamond B_\omega) =
\T(A_\omega \diamond (B_\omega \odot_R C_\omega))   \, .
\end{gather}
($\T$, the  trace per unit volume, is defined in \eqref{tuvintro}.) 
We also recall \cite[Lemma 3.24]{BGKS} that, if $B_{n,\omega}$ is a bounded sequence in
$\K_\infty$ such that $B_{n,\omega}\to B_\omega$ weakly,  then   for all
$A_\omega \in \K_1$ we have 
\begin{align}
 \T (B_{n,\omega} \odot_L A_\omega) \to \T(B_{\omega}\odot_L A_\omega)  \text {\; and \;} \T ( A_\omega\odot_R B_{n,\omega}  ) \to
\T(A_\omega\odot_R B_{\omega}) . \label{lemmaduality2} 
\end{align}
 
Let us set 
\begin{align} \label {Q(0)}
 Q^{(0)} :=  \D_{L }^{(0)} \cap  \D_{R}^{(0)}  =  \left \{A_\omega \in \K_2 \,  | \, {H_\omega(t)}^{\frac{1}{2}}  A_\omega, {H_\omega(t)}^{\frac{1}{2}}  A_\omega^{\ddagger} \in \K_2 \right \}.
\end{align}
where independence of $Q^{(0)}$ in $t$ is justified by Proposition~\ref{DL0_Core}. We also define $	\LL_t$  on $Q(|\LL|)$ by 
\begin{align} 
	\LL_t (A_\omega,B_\omega) := \HH_{L,t}(A_\omega,B_\omega) - \HH_{R,t}(A_\omega,B_\omega)
	\label{L_t}
\end{align}
where definitions \eqref{quadform_HL}-\eqref {QL_t} are used here. Recalling   \eqref{UtsUst}, we now show that, given $B_\omega \in \K_2$, $ \U(t,r) (B_\omega)$ is differentiable in both $t$ and $r$. We state this result in the following proposition.

\begin{proposition} \label{formula59}
Let $A_\omega \in Q(|\LL|)$ and $B_\omega \in Q^{(0)}$. The map $t \mapsto \U(t,r) (B_\omega) \in \K_2$ is differentiable in $\K_2$ in t he following sense: 
\begin{align}
& i \partial_{t} \la \la  A_\omega, \U(t,r) (B_\omega)\ra\ra =   \LL_t (  A_\omega,\U(t,r) (B_\omega )).  \label{eq57}
\end{align}
Similarly,  
\begin{align}
\quad & i \partial_{r} \la \la  B_\omega, \U(t,r) (A_\omega)\ra\ra = -  \LL_r (\U(r,t)  (B_\omega), A_\omega )\, . \label{eq58}
\end{align}
\end{proposition}

\begin{proof}
We first focus on  \eqref{eq57}.  Let us  note that
\begin{align}
& i \partial_{t} \la\la A_\omega , \U(t,r) (B_\omega) \ra\ra   =  i \partial_{t} \la\la  A_\omega, U_\omega (t,r) \odot_L  B_\omega \odot_R U_\omega (r,t) \ra\ra \\
& = \lim_{h \to 0}   \frac{i}{h}  \la\la A_\omega ,   \left (  U_\omega (t +h, r ) - U_\omega(t,r) \right ) \odot_L  B_w \odot_R U_\omega (r,t) \, \ra\ra \label {eq59a} \\ 
& \quad \,   +  \lim_{h \to 0}   \frac{i}{h}  \la\la A_\omega ,  U_\omega (t,r) \odot_L  B_w \odot_R \left ( U_\omega (r, t +h ) - U_\omega(r,t) \right )   \, \ra\ra \, . \label {eq59b}
\end{align}  
Using \eq{centralK2}-\eq{central2}, we can first rewrite \eqref{eq59a} as follows:
\begin{align}
&\lim_{h \to 0} \frac{i}{h} \left \la \left \la A_\omega, \left ( U_\omega( t+h, r) - U_\omega(t,r) \right \} \odot_L B_\omega \odot_R U_\omega (r,t) \, \right \ra \right \ra \\
& = \lim_{h \to 0} \frac{i}{h}   \la   \la \, {\H_L(t_0)}^{\frac{1}{2}} A_\omega , {\H_L(t_0)}^{- \frac{1}{2}} \left \{ U_\omega( t+h, r) - U_\omega(t,r)  \right \} {H_\omega (t_0)}^{- \frac{1}{2}} \odot_L   \nonumber\\
&     \qquad  \qquad \qquad  \odot_L  {H_\omega (t_0)}^{\frac{1}{2}} B_\omega \odot_R U_\omega (r,t) \,  \ra   \ra \\
& = \lim_{h \to 0} \textstyle{\frac{i}{h}} \, \T \left (   {\H_L(t_0)}^{- \frac{1}{2}} \left \{ U_\omega( t+h, r) - U_\omega(t,r)  \right \}  {H_\omega (t_0)}^{- \frac{1}{2}} \odot_L \right.\nonumber \\
&  \left.  \qquad  \qquad \qquad \qquad  \quad \odot_L {H_\omega (t_0)}^{\frac{1}{2}} B_\omega \odot_R U_\omega(r,t) \diamond ({\H_L(t_0)}^{\frac{1}{2}} A_\omega )^{\ddagger}  \right ) \, . \label{eq59c}
\end{align}

We take  the limit  in  \eqref{eq59c} inside the trace per unit volume, using \eq{lemmaduality2}.  Remember that  $\H_L(t)^{- \frac{1}{2}} = H_\omega(t)^{- \frac{1}{2}} \odot_L$,  and note the following reformulation of \eqref{w_partial_t} in Theorem~\ref{WeakYosida}:
\begin{align}
	i\partial_t H(t_0)^{-\frac{1}{2}} U(t,s)H(t_0)^{-\frac{1}{2}} &=  (H(t)^{\frac{1}{2}} H(t_0)^{-\frac{1}{2}})^*  H(t)^{\frac{1}{2}} U(t,s) H(t_0)^{-\frac{1}{2}}. \label{weakpartial_t}%\\
\end{align}
\eqref{eq59c} is now equal to:
\begin{align}
& = \T \left ( ({H_\omega (t)}^{\frac{1}{2}} {H_\omega (t_0)}^{-\frac{1}{2}} )^{*} {H_\omega (t)}^{\frac{1}{2}} U_\omega(t,r) {H_\omega (t_0)}^{-\frac{1}{2}} \odot_L \right . \nonumber \\
 &   \left . \qquad \quad \qquad \qquad  \quad   \odot_L {H_\omega (t_0)}^{\frac{1}{2}} B_\omega \odot_R U_\omega(r,t) \diamond ( {\H_L(t_0)}^{\frac{1}{2}} A_\omega)^{\ddagger}  \right) \\
 & =   \T  \left (   \,   H_\omega(t)^{\frac{1}{2}} \U(t,r) (B_\omega)   \diamond  (  \H_L(t_0)^{\frac{1}{2}} A_\omega  )^{\ddagger}   \odot_R ( H_\omega(t)^{\frac{1}{2}}  H_\omega(t_0)^{- \frac{1}{2}} )^{*}  \right )   \label{eq59d} \\
 & = \T \left( {\H_L(t)}^{\frac{1}{2}} \U (t,r) ( B_\omega) \diamond  ( {\H_L(t)}^{\frac{1}{2}} A_\omega  )^{\ddagger} \right )  \label{eq59e}\\
& =   \la   \la { \H_L(t)}^{\frac{1}{2}}  A_\omega, {\H_L(t)}^{\frac{1}{2}} \U(t,r)(B_\omega)   \ra   \ra \ , \label{eq59f}
\end{align} 
where the equality in \eqref{eq59d} is due to \eqref{central2}, and to go from \eqref{eq59d} to \eqref{eq59e}, we used the fact that
\begin{align}
	H_\omega(t_0)^{- \frac{1}{2}} \odot_L C_\omega    = 	\H_L (t_0)^{- \frac{1}{2}} C_\omega  \in Q^{(0)} \quad \text{with}\quad  C_\omega :=   \H_L(t_0)^{\frac{1}{2}} A_\omega \,  \in \K_2. \label{H=H_L}
\end{align}
Indeed, this shows \eqref{eq59e} since
\begin{align}
	 &   C_\omega^{\ddagger}   \odot_R  (H_\omega(t)^{\frac{1}{2}}  H_\omega(t_0)^{- \frac{1}{2}}  )^{*}  =  ( H_\omega(t)^{\frac{1}{2}}  H_\omega(t_0)^{- \frac{1}{2}} \odot_L C_\omega )^{\ddagger} = {C_\omega}^{\ddagger}   \label{eq59f_2}
\end{align}
 by \eqref{leftrightM},  \eqref{HLtAHtA}, and  \eqref{H=H_L}.

Repeating the above arguments, one can also show that 
\begin{align}\label{eq59h}
  & 	 \lim_{h \to 0}   \frac{i}{h}  \la\la A_\omega ,  U_\omega \odot_L  B_\omega \odot_R \left ( U_\omega (r, t +h ) - U_\omega(r,t) \right )   \, \ra\ra \\
   & \qquad \qquad \qquad  = -   \la   \la { \H_L (t)}^{\frac{1}{2}}  \U_\omega (  B_\omega^{\ddagger}) ,   \,  { \H_L (t)}^{\frac{1}{2}}   A_\omega^{\ddagger}    \ra  \ra.  \notag
\end{align}
With \eqref{eq59f} and \eqref{eq59h}, we get \eqref{eq57}.

The equality \eqref{eq58} now follows from \eqref{eq57}. Indeed, for all $A_\omega, B_\omega \in \K_2$ 
\begin{align}\notag
	 & \la \la  A_\omega, \U(t,r) (B_\omega) \ra\ra =  \la \la  A_\omega, \,U_\omega (t,r) \odot _ L B_\omega \odot_R  U_\omega (r,t)\ra\ra \\\notag
		&  \quad		  =  \left \la \left \la   (B_\omega \odot_R  U_\omega (r,t))^{\ddagger} , \, (U_\omega (r,t) \odot _ L A_\omega)^{\ddagger} \right \ra \right \ra  \\
		&  \quad = \left \la \left\la  B_\omega ^{\ddagger}, \, \{\U_\omega (r,t) ( A_\omega) \}^{\ddagger}  \right\ra\right\ra =  \left \la\left \la  \U_\omega (r,t) ( A_\omega) , \, B_\omega \right\ra\right\ra. \label{eq60}
\end{align}
Hence,   \eqref{eq57}  gives us
\begin{align}\notag
& i \partial {r} \la \la  B_\omega, \U(t,r) (A_\omega) \ra\ra   =	i \partial {r}\la \la  \U (r,t)  (B_\omega) , \, A_\omega \ra\ra   = \overline{	-i \partial {r} \la \la  A_\omega, \U(r,t) (B_\omega) \ra\ra}     \\ 
&  \qquad   \qquad =  \overline {-  \LL_r ( A_\omega , \U(r,t) (B_\omega)  )}  =  -  \LL_r ( \U(r,t) (B_\omega) , A_\omega  ).
\end{align}
\end{proof}

\section {The generalized Liouville equation}
 Let $H_{\omega}$ be the ergodic magnetic Schr\"odinger operator in \eq{H(A,V)omega}-\eq{covintro}.
 With the adiabatic switching of a spatially homogeneous electric field $\El$, the system is described by the  time-dependent Hamiltonian $H_\omega(t)$ as in \eqref{eq:Homegaintro}. By Theorem~\ref{Q_invariant}, the quadratic form domain $Q(H_\omega(t))$ is independent of $t$.

We fix an initial equilibrium state at $t = - \infty$, for which we use the density matrix $\zeta_\omega$. For physical applications, we take $\zeta_\omega = f(H_\omega)$ with $f$
the Fermi-Dirac distribution at inverse temperature $\beta \in (0,\infty]$ and
Fermi energy $E_F \in \R$, that is,  
\begin{equation}\label{defzeta}
{\zeta}_\omega \ = \ \begin{cases}  F^{(\beta,E_F)}_\omega \ := \
\frac{1}{1+\e^{\beta (H_\omega - E_F)}} \, , &
\beta < \infty \, , \\
P^{(E_F)}_\omega \ := \ \chi_{(-\infty,E_F]}(H_\omega) \, , &\beta = \infty \, .
\end{cases}
\end{equation}
The key hypotheses are that $\zeta_\omega$ is real-valued, $\zeta_\omega \ge 0$, and for $k = 1, 2,  \cdots, d$, 
\begin{align}
	[x_k, \zeta_\omega] \in \K_2.
\end{align}
which is implied by  assumption \eqref{key_hypothesis}.   Note that  it follows from \eq{defzeta}  that $H_\omega^\frac{1}{2} \zeta_{\omega}\in \K_{\infty} \cap \K_{2}$ (to see  $H_\omega^\frac{1}{2} \zeta_{\omega}\in  \K_{2}$ use also \cite[Proposition~2.1]{BGKS}).
In particular,  we have  $\zeta_\omega \in Q^{(0)}$.

Recalling  \eqref{Script_G(t)}, we set
\begin{align}\label{zeta(t)}
	 \zeta_\omega(t) = \G(t)(\zeta_\omega) = G(t)\, \zeta_\omega \, G(t)^*,
\end{align}
and note that $\zeta_\omega(t) \in Q^{(0)}$ in view of Propositions~\ref{DL0_Core} and ~\ref{HR=HLdagger}.  

The density matrix $\varrho_\omega(t)$  evolves formally by the Liouville equation \eqref{Liouvilleeq}.
Remembering $\LL_t$   on $Q(|\LL|)$ from \eqref{QL_t} and \eqref{L_t}, we now state the following theorem, which generalizes Theorem~\ref{thmrho}.

\begin{theorem} \label{thm_rho}
There exists a unique $Q(|\LL|)$-valued function $\varrho_\omega(t)$, which solves the Liouville equation \eqref{Liouvilleeq} weakly in $Q^{(0)}$ in the following sense:
\begin{equation} \label{Liouville_K2}
\left\{ \begin{array}{l} i\partial_t \la\la A_\omega, \varrho_\omega(t)\ra\ra = \mathbb{L}_t (A_\omega, \varrho_\omega(t)) \\ 
\lim_{t \to -\infty} \la\la A_\omega, \varrho_\omega(t) \ra \ra = \la\la A_\omega , {\zeta}_\omega \ra \ra
\end{array}
\right. 
\end{equation}
for all $A_\omega \in Q^{(0)}$. 

Moreover, the unique solution  $\varrho_{\omega} (t) \in  Q^{(0)} $ for all $t$ and is given by 
\begin{align}
\varrho_\omega(t) &=\lim_{s \to -\infty}{ \U}(t,s)\left( {\zeta}_\omega
\right) = \lim_{s \to -\infty}{ \U}(t,s)\left( {\zeta}_\omega(s) \right) \label{rho_as_limitUzeta}
\\
&= {\zeta}_\omega(t)
- i
 \int_{-\infty}^t \mathrm{d} r \,\mathrm{e}^{\eta {r_{\! -}}}{ \,\U}(t,r) \left([ \mathbf{E}
\cdot \x, {\zeta}_\omega(r) ]\right).  \label{def_rho3intro}
\end{align}
\end{theorem} 

\begin{lemma} \label{lem_bath}
 Let $A_\omega \in Q(|\LL|)$. Then,
 \begin{align}
    \LL_t (A_\omega, \zeta_\omega(t)) = 0. \label{eq64}
\end{align}
\end{lemma}

\begin{proof} 
Without loss of generality, let $A_\omega = A_\omega^{\ddagger}$ and  consider 
\begin{align}
  \widetilde \zeta_\omega (t) := f_n (\H_L (t)) \zeta_\omega (t) \, \in \D(\H_L (t)), \label{eq65}
\end{align}
where ${f_n}$ be a sequence of measurable, bounded, and compactly supported functions such that $f_n \to \delta$, the Dirac-delta function.

Since ${\H_L(t)}  \widetilde \zeta_\omega (t)$ is self-adjoint, a simple calculation yields
\begin{align}
 &  \HH_{L,t} (A_\omega , \widetilde  \zeta_\omega (t) )    =    \HH_{R,t} (A_\omega , \widetilde  \zeta_\omega (t) ), 
\end{align}
which shows \eqref{eq64} with $\widetilde \zeta_\omega (t)$.   However,  by \eqref{central}, \eqref{centralK2}, and \eqref{lemmaduality2} 
\begin{align}
	&   \la  \la {\H_L(t)}^{\frac{1}{2}} A_\omega , {\H_L(t)}^{\frac{1}{2}} \widetilde \zeta_\omega (t)    \ra \ra = \T \left (   \{  {\H_L(t)}^{\frac{1}{2}} A_\omega   \}^{\ddagger} \diamond  {\H_L(t)}^{\frac{1}{2}} f_n (\H_L (t) ) \zeta_\omega (t) \right )  \nonumber \\
	  & \quad = \T \left (    f_n (\H_L (t))  \left [ {\H_L(t)}^{\frac{1}{2}}  \zeta_\omega (t)    \diamond   \{   {\H_L(t)}^{\frac{1}{2}} A_\omega   \}^{\ddagger} \right ] \right ) \label{eq67}  \\
	  &  \quad \rightarrow \T \left (   {\H_L(t)}^{\frac{1}{2}}  \zeta_\omega (t)    \diamond   \{   {\H_L(t)}^{\frac{1}{2}} A_\omega  \}^{\ddagger}  \right )  =  \la  \la {\H_L(t)}^{\frac{1}{2}} A_\omega , {\H_L(t)}^{\frac{1}{2}}  \zeta_\omega (t)    \ra  \ra\nonumber
\end{align}
as $n \to \infty$. Similarly, one can also show
\begin{align}
 \la   \la {\H_L(t)}^{\frac{1}{2}} \widetilde \zeta_\omega (t), {\H_L(t)}^{\frac{1}{2}}  A_\omega    \ra   \ra \rightarrow   \la   \la {\H_L(t)}^{\frac{1}{2}}  \zeta_\omega (t), {\H_L(t)}^{\frac{1}{2}}  A_\omega   \ra  \ra  
\end{align}
as $n \to \infty$. Together with \eqref{eq67} we have \eqref{eq64} for all $ \zeta_\omega (t)$.    
 \end{proof}
The following lemma plays a crucial role in proving Theorem~\ref{thm_rho}.
\begin{lemma}  \label{crucial_lemma}
We have
$
	 [ \mathbf{E} \cdot \x, {\zeta}_\omega ] \in Q^{(0)}
$.
\end{lemma}

\begin{proof} 

We need to prove  $[ \mathbf{E} \cdot \x, {\zeta}_\omega ]\, \H_c \subset \D ( {H_\omega}^{\frac{1}{2}})$  and   $H_\omega^{\frac{1}{2}}\ [ \mathbf{E} \cdot \x, {\zeta}_\omega ] \in \K_2$.  We will do this  for ${\zeta}_\omega  = P_\omega^{(E_F)},$ since the other case is slightly easier. 
It also suffices to show that $ \left[  x_j, P_\omega^{(E_F)} \right] $ is in $Q^{(0)}$ for each $j=1,2,\ldots,d$.

To show $\left[x_j,P_\omega^{(E_F)}\right]\H_c \subset \D  ({H_\omega}^{\frac{1}{2}})$, we only need to verify   $x_j P_\omega^{(E_F)} \varphi \in \D   ({H_\omega}^{\frac{1}{2}})$ for all $\varphi \in \H_c$ (since $ P_\omega^{(E_F)}x_j \varphi \in \D  ({H_\omega}^{\frac{1}{2}}  )$  for all $\varphi \in \H_c$).
Let $\eta_N$ be a sequence of $C_0^{\infty}(\R)$ functions  such that   $|\eta_N| \le 1$ for all $N$ and $\eta_N = 1$ for  $x_j \in [-N, N]$ and decays to $0$ otherwise. Then, 
\begin{align}
 x_j^{(N)} := x_j \eta_N(x_j)\, \in C_0^{\infty}(\R)
\end{align} 
has the property that for all $N$, $\partial_{x_j} \left|x_j^{(N)}\right| < C$   for some fixed constant $C$.  By assumption \eqref{key_hypothesis} (which  implies  $x_j P_\omega^{(E_F)} \chi_0  \varphi \in \H$), we have that as $N \to \infty$
\begin{align}
	 x_j^{(N)} P_\omega^{(E_F)}    \phi  \to x_j  P_\omega^{(E_F)}   \phi \quad  \text{in} \quad  \H
	 \end{align}
for all $\phi = \chi_0 \varphi \in \H_c$.
 
Recall  $H_+   = H(\A, V_+) =  \Db^*\Db + V_+ $ and note that $\Db x_j^{(N)} P_\omega^{(E_F)} \phi$ is well-defined  since $\left \| \Db  P_\omega^{(E_F)} \right \| = \left \| \Db   { {H_\omega}}^{- \frac{1}{2}}   {H_\omega}^{\frac{1}{2}} P_\omega^{(E_F)}    \right \| < \infty $ and
\begin{align} \label {Dxj}
\Db x_j^{(N)} P_\omega \phi =  { x_j^{(N)}} \Db   P_\omega^{(E_F)} \phi - i \delta_{kj} \textstyle{\frac{\partial} {\partial{x_j}}} \left( x_j^{(N)}\right)    P_\omega^{(E_F)} \phi \ , 
\end{align} 
where $\delta_{kj}$ is the vector consisting of $1$ in the $j$-th entry (i.e., when $k=j$) and $0$ elsewhere.
Since $x_j^{(N)}$ is compactly supported, $  {V_+}^{\frac{1}{2}} x_j^{(N)}  P_\omega^{(E_F)} \phi $ is well-defined as well, and this shows $x_j^{(N)}  P_\omega \phi \in \D  (  {H_+}^{\frac{1}{2}})$.

In \eqref{HxjN}-\eqref{inequality_lemma},
we use $P_\omega = P_\omega^{(E_F)}$,  $   \left(x_j^{(N)}\right) ^{'} =  \textstyle{\frac{\partial} {\partial{x_j}}}  x_j^{(N)}  $, and $ x_j^{(N,M)} =   \left( x_j^{(N)} - x_j^{(M)}\right)$. Given $\phi \in \H_c$, note that
\begin{align} \label{HxjN}
 &\left \|  { H_+}^{\frac{1}{2}}  x_j^{(N,M)}   P_\omega \phi\right \|^2    \\
 & \quad =    \la \Db   x_j^{(N,M)}  P_\omega \phi, \Db   x_j^{(N,M)}  P_\omega \phi   \ra  +   \la   x_j^{(N,M)}   P_\omega \phi, V_+   x_j^{(N,M)}  P_\omega \phi  \ra   \\
 & \quad  =   \la  {H_+}^{\frac{1}{2}}  P_\omega \phi, {H_+}^{\frac{1}{2}}    \left(  x_j^{(N,M)} \right)^2   P_\omega \phi  \ra  - i     \la       \left (  x_j^{(N,M)}   \right)^{'} P_\omega \phi, \Db_j     x_j^{(N,M)}    P_\omega \phi  \ra \label{equality_2nd} \\
 & \quad \le  \, \left \|  {H_+}  P_\omega \phi  \right \| \left \|   \left (  x_j^{(N,M)} \right )^2  P_\omega \phi   \right \|  + \widetilde C \left \|  \Db_j  P_\omega \phi     \right \| \left \|     x_j^{(N,M)}     P_\omega \phi     \right \| \rightarrow    0  \label {inequality_lemma}
\end{align} 
as $N, M \to \infty$. Here, $\widetilde C = \sup_{N,M}\left | \left (  x_j^{(N,M)} \right )^{'}   \right |$ and equality in \eqref{equality_2nd} is justified by \eqref{Dxj}. Thus,  ${H_+}^{\frac{1}{2}}x_j^{(N)}  P_\omega^{(E_F)} \phi$ is  a  Cauchy sequence in $\H$, and its limit  ${H_+}^{\frac{1}{2}}x_j  P_\omega^{(E_F)} \phi$  is in $\H$ for a.e. $\omega$. It also follows  that ${H_\omega }^{\frac{1}{2}}x_j  P_\omega^{(E_F)} \phi \in \H$ for a.e.\ $\omega$ by Theorem~\ref{Q_invariant}.

We now turn to the claim ${H_\omega}^{\frac{1}{2}} \left[x_j ,  P_\omega^{(E_F)} \right] $ is in $\K_2$. For this, it suffices to verify 
\begin{align}
\E \left \{ \left \| {H_\omega}^{\frac{1}{2}}x_j   P_\omega^{(E_F)} \chi_0     \right \|_2^2  \right \} < \infty.
\end{align}
Let $\{\phi_n\}$ be an orthonormal basis for $\chi_0 \H$. 
By the arguments that led to \eqref{inequality_lemma},  we   have
\begin{align}
 &\E \left \{ \left \| {H_\omega}^{\frac{1}{2}}x_j   P_\omega^{(E_F)} \chi_0     \right \|_2^2  \right \}  = \E \sum_n \left \la  H_\omega^{\frac{1}{2}}x_j   P_\omega^{(E_F)} \phi_n ,  {H_\omega}^{\frac{1}{2}}x_j   P_\omega^{(E_F)}\phi_n \right \ra      \\
 & \quad \le \E  \sum_n \left \{  \left \|  {H_\omega }  P_\omega^{(E_F)} \phi_n  \right \| \left \|     x_j ^2  P_\omega^{(E_F)} \phi_n  \right \|  +   C \left \|  \Db_j  P_\omega^{(E_F)} \phi_n     \right \| \left \|  x_j    P_\omega^{(E_F)} \phi_n    \right \|   \right \} \nonumber \\
 & \quad  \le    C^{'} \E  \left \{  \left \|  {H_\omega }  P_\omega^{(E_F)} \chi_0  \right \|_2 ^2 +  \left \|     x_j ^2  P_\omega^{(E_F)}\chi_0 \right \|_2 ^2  +   \left \|  \Db_j  P_\omega^{(E_F)} \chi_0   \right \|_2  ^2 +  \left \|  x_j    P_\omega^{(E_F)} \chi_0   \right \| _2^2   \right \}  \nonumber
\end{align}
which is bounded by \eqref{key_hypothesis} and the fact that ${H_\omega}  P_\omega^{(E_F)} $ and $\Db_j  P_\omega^{(E_F)} $ are in $\K_2$  (see \cite[Lemma 5.4]{BGKS}). Here, $C^{'}$ is constant, and this completes the proof.
\end{proof}

 With Lemmas~\ref {lem_bath} and ~\ref{crucial_lemma}, we now prove Theorem~\ref{thm_rho}.

\begin{proof} [Proof of Theorem~\ref{thm_rho}.]
By \eqref{partial_G(t)} and Proposition \ref {formula59}, we first note that,  
 given $A_\omega \in  Q^{(0)}$, we have
\begin{align}\notag
	 &   i\partial_s \la\la A_\omega, \U(t,s)(\zeta_\omega (s))\ra\ra \\
 & \quad = -   \LL_s (\U(s,t)  (A_\omega), \zeta_\omega (s) )  -  \la\la A_\omega, \U(t,s)([\El(s)\cdot \x, \zeta_\omega(s)]) \ra\ra  \label{partial_s_rho} \\
  & \quad = - \la\la A_\omega, \U(t,s)([\El(s)\cdot \x, \zeta_\omega(s)]) \ra\ra, \notag
\end{align}
where  the  equality \eqref{partial_s_rho} comes from  Lemma~\ref{lem_bath}.  Hence,   
 \begin{align}
	&\la\la A_\omega,   \U(t,s)(\zeta_\omega (s))\ra\ra  =\la\la A_\omega,   \zeta_\omega(t)\ra\ra - i \int_s^t {\mathrm{d}r\, \la\la A_\omega, \U(t,r)([\El(r)\cdot \x, \zeta_\omega(r)])\ra\ra} \label{eq70}
 \end{align}
which shows  \eqref{def_rho3intro}. The second equality in \eqref{rho_as_limitUzeta}  follows from  the strong continuity of $\G(t)$ .  
Since  $ {\U}(t,r) \left(\left[{\El}\cdot \x,{\zeta}_\omega(s)\right]
\right) \in \K_2$
by \eqref{key_hypothesis} and Proposition~\ref{propUomega}, we have\begin{align}\notag
		&  \la\la A_\omega,   \varrho_\omega(t)\ra\ra = \lim_{s \to -\infty} \la\la A_\omega,  \U(t,s)(\zeta_\omega (s))\ra\ra  \\
		& \quad = \la\la A_\omega,   \zeta_\omega(t)\ra\ra - i \int_{-\infty}^t {\mathrm{d}r\, \la\la A_\omega, \U(t,r)([\El(r)\cdot \x, \zeta_\omega(r)])\ra\ra}. \label{eq71}
\end{align}

We now claim that
\begin{align}
	{H_\omega(t)}^{\frac{1}{2}} \U(t,r)([\El(r)\cdot \x, \zeta_\omega(r) ]) \in \K_2 \label{HUCinK2}. 
\end{align}
First,  note  that
\begin{align}\notag
	& {H_\omega(t)}^{\frac{1}{2}} \U(t,r)([\El(r)\cdot \x, \zeta_\omega(r) ]) \\	
& =  \mathrm{e}^{\eta r_-} {H_\omega(t)}^{\frac{1}{2}} U_\omega (t,r) H_\omega(r)^{- \frac{1}{2}}\odot_L  H_\omega(r)^{ \frac{1}{2}} [\El \cdot \x, \zeta_\omega (r) ] \odot_R U_\omega (r,t)\\
& =  \mathrm{e}^{\eta r_-}  W_\omega(t,r) \odot_L \G(r) \left ( {H_\omega}^{ \frac{1}{2}} [\El \cdot \x, \zeta_\omega ] \right ) \odot_R U_\omega (r,t).\notag
\end{align}
Since $\G(r) $ as in \eqref{Script_G(t)} is a unitary map on $\K_2$,  we have
$ \G(r) \left ( {H_\omega}^{ \frac{1}{2}} [\El \cdot \x, \zeta_\omega ] \right ) \in \K_2$
by Lemma~\ref{crucial_lemma}.  Since  $W_\omega(t,s)$ is uniformly bounded, \eqref {HUCinK2} follows from Proposition~\ref{propUomega}. 

 It also  follows from \eqref{HUCinK2} that
\begin{align} \label{eq75}
	  \int_s^t {\mathrm{d}r\, {H_\omega(t)}^{\frac{1}{2}} \U(t,r)([\El(r)\cdot \x, \zeta_\omega(r) ])} \in \K_2 \, ,
\end{align}
where we used the Bochner integral above. Moreover,    
\begin{align}
	&  \int_s^t {\mathrm{d}r \, {H_\omega(t)}^{\frac{1}{2}} \U(t,r)([\El(r)\cdot \x, \zeta_\omega(r) ])} 
= {H_\omega(t)}^{\frac{1}{2}} \int_s^t {\mathrm{d}r \, \U(t,r)([\El(r)\cdot \x, \zeta_\omega(r) ])}\, , \label{eq77}
\end{align}
by a similar argument following  \cite[Eq.~(5.29)]{BGKS}.

We now continue from \eqref{eq71}. Recalling Proposition~\ref{formula59},  Lemma~\ref{lem_bath}, \eqref{partial_G(t)}, and \eqref{zeta(t)}, letting $A_\omega \in Q^{(0)}$ we have 
\begin{align}
&	i \partial_t \la\la A_\omega,   \varrho_\omega(t)\ra\ra \\ \notag
&= - \la\la A_\omega,  [\El(t)\cdot \x, \zeta_\omega(t)]  \ra\ra  +  \la\la A_\omega, \,  \U(t,t) ([\El(t)\cdot \x, \zeta_\omega(t)]) \ra\ra  \\\notag
&\qquad \, -  i \int_{-\infty}^t {\mathrm{d}r  \left \{\, i \partial_t \,\la\la A_\omega,  \U(t,r)([\El(r)\cdot \x, \zeta_\omega(r)])\ra\ra \right \}}   \\
& = - i \int_{-\infty}^t {\mathrm{d}r   \la   \la{H_\omega(t)}^{\frac{1}{2}} A_\omega , \, {H_\omega(t)}^{\frac{1}{2}} \U(t,r)([\El(r)\cdot \x, \zeta_\omega(r) ])  \ra   \ra } \label{eq80}\\
& \qquad + i \int_{-\infty}^t {\mathrm{d}r   \la   \la {H_\omega(t)}^{\frac{1}{2}}\U(t,r)([\El(r)\cdot \x, \zeta_\omega(r) ]^{\ddagger}) , \, {H_\omega(t)}^{\frac{1}{2}}   A_\omega^{\ddagger}  \ra   \ra} \label{eq81} \\
& =  \widetilde \LL_t \left (  A_\omega,  - i \int_{-\infty}^t {\mathrm{d}r \, \U(t,r)([\El(r)\cdot \x, \zeta_\omega(r) ])}\right) \label {eq84}\\
& =  \widetilde \LL_t \left (  A_\omega, \, \zeta_\omega (t) - i \int_{-\infty}^t {\mathrm{d}r \, \U(t,r)([\El(r)\cdot \x, \zeta_\omega(r) ])}\right)  = \widetilde \LL_t \left (  A_\omega, \, \varrho_\omega (t) \right ) \label{eq87}.
\end{align}
where to go from \eqref{eq80}-\eqref{eq81} to \eqref{eq84}  we used  \eqref{eq75}, \eqref{eq77}, and the definition of $\widetilde \LL_t$  in \eqref{L_tilde}.

 Since $- i \int_{-\infty}^t {\mathrm{d}r \, \U(t,r)([\El(r)\cdot \x, \zeta_\omega(r) ])} \in Q^{(0)}$
by \eqref{eq75} and \eqref{eq77}, we can replace  $\widetilde \LL_t$ in \eqref{eq84}-\eqref{eq87} with $\LL_t$ by \eqref{HLAHA}. This shows that $\varrho_\omega (t)$, given in \eqref{def_rho3intro}, is  a solution to the Liouville equation \eqref{Liouville_K2} (and \eqref{GenLiouvilleeq}) for all $A_\omega \in Q^{(0)}$.

It remains to show the solution  $\varrho_\omega (t)$ is unique in $\K_2$. Let $v_\omega(t)$ be a solution of \eqref{Liouville_K2} with $\zeta_\omega = 0$, then it suffices to show that $v_\omega(t) = 0$ for all $t$. By \eqref{eq60}, which states that for $A_\omega, B_\omega \in \K_2$
\begin{align}
	\la \la  A_\omega, \U(t,r) (B_\omega) \ra\ra = \la \la \U(r,t) (A_\omega),  B_\omega \ra\ra \ ,
\end{align}
we have that for $ A_\omega \in Q^{(0)}$ and $v_\omega(t) \in Q(|\LL|)$,  
\begin{align}
	& i\partial_t 	\la \la  A_\omega, \U(s,t) (v_\omega(t) ) \ra\ra  = - \LL_r ( \U(t,s) (A_\omega) ,v_\omega (t) ) +   \LL_r ( \U(t,s) (A_\omega) ,v_\omega (t) ) = 0\, 
\end{align}
by Proposition~\ref{formula59}. Hence, letting $t = s$, we conclude that   for all $A_\omega \in Q^{(0)}$,
 \begin{align}
	 	\la \la  A_\omega, \U(s,t) (v_\omega(t) ) \ra\ra = 	\la \la  A_\omega, \U(s,s) (v_\omega(s) ) \ra\ra = 	\la \la  A_\omega, v_\omega(s)  \ra\ra \ .
\end{align}
 This shows  $\la \la  A_\omega, v_\omega(t)  \ra\ra = \la \la  A_\omega, \U(t,s)(v_\omega(s))  \ra\ra$, and letting $s \to -\infty$ we see that $\la \la  A_\omega, v_\omega(t)  \ra\ra = 0 $   for all $t$ and for all $A_\omega \in Q^{(0)}$. Since $Q^{(0)}$ is dense in $\K_2$,  $v_\omega(t) = 0$, and this completes the proof.
\end{proof}

%%%%%%%%%%%%%%%%%%%%%%%%%%%%%%%%%%%%%%%%%%%%%%%%%%%%%%%%%%%%%%%%%%%%%%%%%%%%%%%%


\begin{thebibliography}{BGKS}




\bibitem[BGKS]{BGKS} { Bouclet, J.M.,   Germinet, F., Klein, A.,  Schenker, J.H.} {\it Linear response theory for magnetic Schr\"odinger operators in
    disordered media},
    J. Funct. Anal. 226, 301-372 (2005).



\bibitem[CFKS]{CFKS}  Cycon, H.L.,   Froese, R.G.,  Kirsch, W., Simon, B.:
 {\em Schr\"odinger operators}. Heidelberg: Springer-Verlag, 1987.
 
 
 
\bibitem[F]{F} Faris, W.G.: \emph{Self-Adjoint Operators}.
Lecture Notes in Mathematics {\bf 433}. Springer-Verlag, 1975.


\bibitem[GK1]{GKdecay}  Germinet, F.,  Klein, A.: Operator kernel estimates for  functions of generalized Schr\"odinger operators.
Proc. Amer. Math. Soc. 131,  911-920  (2003).


\bibitem[GK2]{GKsudec} Germinet, F., Klein, A.: New characterizations of the region of complete localization for random Schr\"odinger operators.  J. Stat. Phys. \textbf{122}, 73-94 (2006)



\bibitem[K]{Ky} Kisy$\acute{\text{n}}$ski, J.:
{Sur les op$\acute{\text{e}}$rateurs de Green des probl$\grave{\text{e}}$ms de Cauchy abstraits}.
Studia Mathematica, T.~{\bf XXIII}, 285-328 (1964).



\bibitem[LS]{LS} { H. Leinfelder, C.G.  Simader},
{\it Schr\"odinger operators with singular magnetic potentials},
Math. Z. 176, 1-19 (1981).



\bibitem[RS1]{RS1}  Reed, M.,  Simon, B.:  \emph{Methods of Modern
Mathematical Physics I: Functional Analysis}, revised and enlarged edition.
 Academic Press, 1980.

\bibitem[RS2]{RS2}  Reed, M.,  Simon, B.:  \emph{Methods of Modern
Mathematical Physics II: Fourier Analysis, Self-Adjointness}.
Academic Press, 1975.

\bibitem[S1]{Si3}  Simon, B.: \emph{Quantum mechanics for Hamiltonians defined as quadratic forms}. Princeton, N.J., Princeton University Press, 1971.

\bibitem[S2]{Si79}  Simon, B.: Maximal and minimal Schr\"odinger forms.
J. Operator Theory \textbf{1}, 37-47 (1979).
 

\bibitem[Y]{Y} Yosida, K.:  \emph{Functional Analysis, 6th edition.}
Springer-Verlag, 1980.

\end{thebibliography}
\end{document}